\documentclass[reqno,eucal]{amsart}

\usepackage[mathscr]{eucal}
\usepackage{graphicx}
\usepackage{xcolor}
\usepackage{amsmath,amsfonts,amssymb,amsthm,amscd}
\usepackage{longtable}
\usepackage{array}
\usepackage{float}
\usepackage{enumitem}
\usepackage{placeins}
\usepackage{subcaption}
\newtheorem{theorem}{Theorem}
\newtheorem{lemma}[theorem]{Lemma}

\newtheorem{corollary}[theorem]{Corollary}

\theoremstyle{definition}

\newtheorem{example}[theorem]{Example}

\newcommand{\Z}{{\mathbb Z}}
\newcommand{\R}{{\mathbb R}}
\newcommand{\C}{{\mathbb C}}

\newcommand{\I}{{\mathrm i}}
\newcommand{\e}{{\mathrm e}}
\newcommand{\od}{{\mathcal O}}

\begin{document}

\begin{center}
{\LARGE\bfseries
Functional relations in renormalization group methods\\
for a class of ordinary differential equations
\par}

\vspace{2.5em}

Atsuo Kuniba$^{1}$ and Rurika Motohashi$^{2}$

\vspace{1.5em}

University of Tokyo, Komaba, Tokyo 153-8902, Japan\\
$^{1}$Graduate School of Arts and Sciences,\quad
$^{2}$Department of Integrated Sciences
\end{center}

\vspace{2em}

\noindent\textbf{Abstract.} 
We develop a renormalization group (RG)-based perturbation scheme
for a class of ordinary differential equations,
including first-order systems with semisimple or nilpotent linear parts,
as well as scalar higher-order equations.
The key observation is that the secular coefficients
arising in naive perturbation theory satisfy an exact functional relation.
This yields, in a unified manner, several fundamental features of the RG method:
the renormalized amplitudes satisfy a closed functional relation
with a group-like structure,
the RG equation governing their slow dynamics is obtained directly,
the absence of secular terms is ensured to all orders,
and the relation between bare and renormalized amplitudes admits an explicit inversion.
The results extend earlier ones for second-order scalar equations.

\vspace{1.5em}

\section{Introduction}\label{sec:int}

Perturbation theories for ordinary differential equations (ODEs)
form a broad and well-developed subject, encompassing multiple time scale analysis,
normal form theory, exact WKB methods and related approaches
(cf.~\cite{BO,KT,Kuz,M}).
Among these, the renormalization group (RG) method provides
a systematic framework for eliminating secular terms and extracting
effective long-time dynamics from perturbative expansions
\cite{CGO,C1,C2,DHHJK,EFK,K,Kh}.

In this paper, we formulate an RG-based perturbation scheme for ODEs of the form
\begin{align}\label{int1}
\frac{dy}{dt} = \I M y + \varepsilon V(\varepsilon, \e^{\pm \I t}, y),
\end{align}
where $\I = \sqrt{-1}$, $y$ and $V$ are $n$-dimensional vectors, and $V \in \C[\varepsilon, \e^{\I t}, \e^{-\I t}, y]$
is polynomial in the components of $y$.
The matrix $M \in \mathrm{Mat}(n,\mathbb{Z})$ is assumed to be either semisimple or nilpotent
(precise definitions are given in \eqref{eq1} and \eqref{eq2}).
We also treat scalar $N$th-order ODEs of the form \eqref{deq0}.

As is well known, naive perturbation around $\varepsilon=0$ typically leads to resonance,
resulting in Fourier modes whose coefficients grow polynomially in time.
The RG method removes such secular terms by introducing renormalized amplitudes,
whose slow evolution is governed by the RG equation.

While various implementations of the RG method have been developed,
they are typically described in procedural terms,
and the structural mechanism underlying their success,
particularly at the level of all orders in $\varepsilon$,
is not always transparent.
A key observation of the present work is that the secular coefficients,
constructed as formal power series in $\varepsilon$ by an unambiguous perturbative scheme,
satisfy an exact functional relation
(Corollaries \ref{co:p=p3}, \ref{cor:fun}, and \eqref{pp32}),
which makes the underlying picture explicit.
From this relation, several fundamental features of the RG method follow
in a unified manner:

(i) The renormalized amplitudes $\mathscr{A}$,
defined as a special class of secular coefficients,
satisfy a closed functional relation among themselves,
revealing a group-like structure
(\eqref{AA3}, \eqref{AA32}, \eqref{fun2}).

(ii) The RG equation governing the slow dynamics of $\mathscr{A}$
is obtained directly from this functional relation
(\eqref{rg3}, \eqref{rg32}, \eqref{rgeq2}).

(iii) The absence of secular terms in the renormalized expansion
becomes manifest at all orders in $\varepsilon$
(\eqref{rs3}, \eqref{rs2}).

(iv) The relation between the bare amplitudes $A$
and the renormalized amplitudes $\mathscr{A}$ can be inverted explicitly,
yielding $A=A(\mathscr{A})$
(\eqref{inv3}, \eqref{inv4}, \eqref{inv1}).

\medskip
The functional relation for secular coefficients,
together with its consequences described above,
was first observed in \cite{K} for a class of second-order scalar ODEs.
In particular, that setting already covers a number of classical examples,
such as the Van der Pol, Mathieu, Duffing, and Rayleigh equations,
which have often served as benchmarks for the RG method.
The present work extends this structure to a broader class of systems,
including first-order systems with semisimple or nilpotent linear parts,
as well as higher-order scalar equations, and therefore encompasses
these classical cases as well.

The rest of the paper is organized as follows.
Sections \ref{sec:4} and \ref{sec:5} deal with first-order systems
with semisimple and nilpotent $M$, respectively.
Section \ref{sec:2} treats scalar $N$th-order ODEs.
Representative examples are also provided to demonstrate these features.
Section \ref{sec:c} contains concluding remarks.
Appendix \ref{app:n} presents numerical plots for Example \ref{ex:CD}.

Although the derivation of the functional relation and its consequences in this paper
is essentially elementary, it brings out an underlying structure that becomes transparent
at the level of formal power series in $\varepsilon$. Accordingly, no claims are made
regarding convergence or the accuracy of truncated expansions.

\section{First-Order ODE System with a Semisimple Coefficient Matrix}\label{sec:4}

We consider a system of first-order ODEs of the form
\begin{subequations}\label{eq1}
\begin{align}
\frac{dy}{dt} &= \I M y + \varepsilon V(\varepsilon,\e^{\pm \I t},y),
\label{yeq}\\
y= &\begin{pmatrix}y_1 \\ \vdots \\ y_n \end{pmatrix},
\quad
M=
\begin{pmatrix}m_1 &\cdots & 0
\\
\vdots & \ddots & \vdots
\\
0 &\cdots & m_n 
 \end{pmatrix},
 \quad 
V(\varepsilon,\e^{\pm \I t},y) 
=\begin{pmatrix}V_1(\varepsilon,\e^{\pm \I t}, y) 
\\ \vdots \\ V_n(\varepsilon,\e^{\pm \I t},y) \end{pmatrix}.
\label{ymf}
\end{align}
\end{subequations}
Throughout this section and the next, we will consider several equations 
formulated in terms of $n$-dimensional column vectors as in (\ref{ymf}).
However, for notational convenience, 
we will not strictly distinguish between row and column vectors
when referring to them in inline notation.
For instance,  the first relation in (\ref{ymf}) will be referred to as 
$y=(y_1,\ldots, y_n)$.

The matrix $M$ is diagonal with integer entries $m_1, \ldots, m_n \in \Z$, 
which are not assumed to be pairwise distinct.
For each $1 \le j \le n$, the function $V_j(\varepsilon, \e^{\pm \I t}, y)$ 
is a polynomial in $\varepsilon, y_1,\ldots, y_n$ and a
Laurent polynomial in $\e^{\I t}$,  i.e., 
$V_j(\varepsilon,\e^{\pm \I t}, y) \in 
\C[\varepsilon, \e^{\I t}, \e^{-\I t}, y_1, \ldots, y_n]$.
We do not assume $V(\varepsilon, \e^{\pm \I t}, y=0)=0$.

Switching from $y_j$ to $\tilde y_j = \e^{-\I m_j t} y_j$, the equation
\eqref{yeq} becomes
$d\tilde y/dt = \varepsilon\, \tilde V(\varepsilon, \e^{\pm \I t}, \tilde y)$,
where
$\tilde V_j(\varepsilon, \e^{\pm \I t}, \tilde y)
= \e^{-\I m_j t} V_j(\varepsilon, \e^{\pm \I t}, \e^{\I M t} \tilde y)$
satisfies the same condition stated above.
Thus one may actually assume $M=0$ without loss of generality.
For convenience, however, we keep them in the description.

\subsection{Naive perturbation}\label{sss:np1}
We set
$y= (y_1(\varepsilon,t), \ldots, y_n(\varepsilon,t))$,
and seek a formal power series solution of the form
\begin{align}\label{yf3}
y_j(\varepsilon,t) = \sum_{k\in \Z_{\ge 0}}\varepsilon^k\sum_{m \in \Z} f^{(k)}_{j,m}(t)\e^{\I mt},
\qquad f^{(k)}_{j,m}(t)\;\; \text{a polynomial in $t$}.
\end{align}
This corresponds to the expansion $y=\sum_{k\ge 0} \varepsilon^k y^{(k)}(t)$, i.e., 
\begin{align}\label{eexp}
\begin{pmatrix}y_1(\varepsilon,t) \\ \vdots \\ y_n (\varepsilon,t) \end{pmatrix}
= \begin{pmatrix}y_1^{(0)}(t) \\ \vdots  \\ y_n^{(0)}(t)  \end{pmatrix}
+ \varepsilon \begin{pmatrix}y_1^{(1)}(t) \\ \vdots \\  y_n^{(1)}(t)  \end{pmatrix}
+ \cdots,
\qquad y^{(k)}_j(t) = \sum_{m \in \Z} f^{(k)}_{j,m}(t)\e^{\I mt}.
\end{align}
Substitution of this into (\ref{yeq}) yields the following equations at each order of 
$\varepsilon$:
\begin{align}
\frac{dy^{(0)}}{dt} &= \I M y^{(0)},
\label{yeq0}
\\
\frac{dy^{(1)}}{dt} &= \I M y^{(1)}+ V(0,\e^{\pm \I t}, y^{(0)}),
\nonumber\\
& \vdots
\nonumber\\
\frac{dy^{(k)}}{dt} &= \I M y^{(k)}+ \Biggl[
V(\varepsilon, \e^{\pm \I t}, \sum_{0 \le j <  k}\varepsilon^j y^{(j)})
\Biggr]_{\varepsilon^{k-1}}.
\label{yeqk}
\end{align}
Here, and in what follows, $[X]_{\alpha^k}$ denotes the coefficient of $\alpha^k$ in the
expansion of $X$ as a formal series in $\alpha$.
The general solution to the first equation (\ref{yeq0}) is given by 
$y^{(0)}(t)= \e^{\I M t}A$ for an arbitrary constant vector $A =(A_1,\ldots, A_n)$, 
i.e.,
 \begin{align}\label{shokou}
 \begin{pmatrix}y_1^{(0)}(t) \\ \vdots  \\ y_n^{(0)}(t)  \end{pmatrix}
 =\begin{pmatrix}A_1\e^{\I m_1 t} \\ \vdots  \\ A_n \e^{\I m_n t} \end{pmatrix}.
 \end{align}
Given $y^{(0)}, \ldots, y^{(k-1)}$, 
a solution to the order $\varepsilon^k$ equation (\ref{yeqk}) 
is obtained by adding a special solution of the full inhomogeneous equation and 
a general solution of the homogeneous part,
$dy^{(k)}/dt=\I M y^{(k)}$.
The formal power series solutions thus constructed is {\em uniquely} determined by imposing the following condition:
\begin{align}\label{ycon3}
\begin{pmatrix}y_1 \\ \vdots \\ y_n \end{pmatrix}
= \begin{pmatrix} (A_1+ \od(\varepsilon t)) \e^{\I m_1 t}
+ \sum_{m \neq m_1}\od(\varepsilon t^0)\e^{\I m t}
 \\ \vdots \\
 (A_n+ \od(\varepsilon t)) \e^{\I m_n t}
+ \sum_{m \neq m_n}\od(\varepsilon t^0)\e^{\I m t}
 \end{pmatrix}.
 \end{align}
 Here and in what follows, $X = \mathcal{O}(\alpha)$ means that  
$X / \alpha$ is a formal power series whose terms remain finite as $\alpha \to 0$.
The condition \eqref{ycon3} implies that, in the above-mentioned general
solution, the coefficient $[y^{(k)}_j(t)]_{\mathrm{e}^{\mathrm{i}m_j t}}$
must be a polynomial in $t$ {\em without} a constant term whenever $k\ge 1$.
For $[y^{(k)}_j(t)]_{\mathrm{e}^{\mathrm{i}mt}}$ with $m\notin\{m_1,\ldots,m_n\}$,
constant (i.e., $t^0$) terms may occur also for $k\ge 1$.

Let  $Y(\varepsilon,t,A)= (Y_1(\varepsilon,t,A), \ldots, Y_n(\varepsilon,t,A))$ 
denote the unique solution thus constructed, and define  
 $P_{j,m}(\varepsilon,t,A)$ to be the coefficient in the expansion 
 \begin{align}\label{Yv3}
 Y(\varepsilon,t,A)= \begin{pmatrix}Y_1(\varepsilon,t,A) \\ \vdots \\ Y_n(\varepsilon,t,A) \end{pmatrix}
 = \sum_{m \in \Z}
 \begin{pmatrix}P_{1,m}(\varepsilon, t,A) \\ \vdots \\ P_{n,m}(\varepsilon,t,A) \end{pmatrix}
 \e^{\I mt}.
 \end{align}
 When no confusion is likely, we write this more concisely as
$Y  = \sum_{m \in \Z} P_m(\varepsilon,t,A)\e^{\I mt}$,
where $P_m(\varepsilon,t,A)$ denotes the vector 
$(P_{1,m}(\varepsilon,t,A), \ldots, P_{n,m}(\varepsilon,t,A))$.
We refer to each $P_{j,m}(\varepsilon,t,A)$ as a {\em secular coefficient}.
By definition,  $P_{j,m}(\varepsilon,t,A)$ is a formal power series in $\varepsilon$ and $t$ such that 
$[P_{j,m}(\varepsilon,t,A)]_{\varepsilon^k}$ is a polynomial in $t$ for any $k$.
By construction, it has the following behavior around $\varepsilon=0$:  
\begin{align}\label{pa3}
P_{j,m}(\varepsilon,t,A) = A_j \delta_{m,m_j} + \od(\varepsilon t^{\delta_{m,m_j}})
\quad 
\text{as a formal power series in $\varepsilon$ and $t$}.
\end{align}
The special case of $P_{j,m_j}(\varepsilon,t,A)$
corresponding to $m = m_j$
is called the {\em resonant secular coefficient}, and plays an important role. 
For instance, (\ref{pa3}) implies the relation:
\begin{align}\label{pa32}
A_j = P_{j, m_j}(\varepsilon, 0, A).
\end{align}

Owing to the recursive construction of $y^{(0)}, y^{(1)}, \ldots$ along 
the equations (\ref{yeq0})--(\ref{yeqk}), the secular coefficients acquire increasingly 
higher order power of $\varepsilon$ as $|m|$ becomes large.
Namely, the following property holds for $1 \le j \le n$.
\begin{align}\label{ped3}
P_{j,m}(\varepsilon,t,A) = \mathcal{O}(\varepsilon^{d_{j,m}}),
\quad d_{j,m} \rightarrow \infty\; \text{as }\; |m| \rightarrow \infty.
\end{align}

\subsection{Functional equation for secular coefficients}

From the construction in Section \ref{sss:np1}, we have the following lemma.
\begin{lemma}\label{le:ren3}
Let $s$ and $A=(A_1,\ldots, A_n)$ be arbitrary parameters.
The formal power series $y=y(\varepsilon,t)$ of the form (\ref{yf3}) 
that satisfies the conditions (i), (ii), and (iii) below is unique, 
and coincides with $Y(\varepsilon,t,A)$ in (\ref{Yv3}).
\begin{align*}
&(\I)\;   \text{$y(\varepsilon,t)$ satisfies the differential equation (\ref{yeq})}, \\
&(\I\I)\;   y(0,t) =  (A_1\e^{\I m_1t}, \ldots, A_n\e^{\I m_nt}), \\ 
&(\I\I\I)\; 
[y_j(\varepsilon,t)]_{\e^{\I m_jt}}|_{t=s} = P_{j,m_j}(\varepsilon,s,A)\quad  (1 \le j \le n).
\end{align*}
\end{lemma}

\begin{lemma}\label{le:ren2}
Let $s$ and $B=(B_1,\ldots, B_n)$ be arbitrary parameters.
Then the formal power series
\begin{align}\label{yB}
\sum_{m\in \Z}
 \begin{pmatrix}P_{1,m}(\varepsilon, t-s,B) \\ \vdots \\ P_{n,m}(\varepsilon,t-s,B) \end{pmatrix}
 \e^{\I mt}
 \end{align}
 is also a solution to the equation (\ref{yeq}).
\end{lemma}
\begin{proof}
The nontrivial point is that shifting $t$ to $t-s$ inside 
the coefficients $P_{j,m}$ while leaving the exponential factor 
$\e^{\I mt}$ (i.e., not replacing it with $\e^{\I m (t - s)}$),  
preserves the validity of the solution.
To justify this, regard the expression 
(\ref{Yv3}) as a forma Laurent series in $\e^{\I t}$.

By substituting $Y = \sum_{m \in \Z} P_m(\varepsilon,t,A)\e^{\I mt}$ into 
(\ref{yeq}) and extracting the coefficient of $\e^{\I mt}$, we obtain 
the following infinite system of equations for the secular coefficients:
\begin{align}\label{peq3}
\frac{\partial P_m(\varepsilon,t,A)}{\partial t}
+ \I m P_m(\varepsilon,t,A) = \I MP_m(\varepsilon,t,A)
+ \left[\varepsilon
V\bigl(\varepsilon, \e^{\pm \I t}, \sum_{l\in \Z}P_l(\varepsilon,t,A)\e^{\I lt}
\bigr)\right]_{\e^{\I mt}}
\quad (m \in \Z).
\end{align}
In general, the last term in the RHS involves infinite sums.
However, thanks to (\ref{ped3}), these sums are actually finite at each order  
in $\varepsilon$, and thus make sense as a formal power series in 
$\varepsilon$.
Note that the system (\ref{peq3}) is fully autonomous:
The variable $t$ appears only through $\{P_l(\varepsilon, t,A)\mid l \in \Z\}$,
regardless of whether the original equation (\ref{yeq}) is autonomous or not.
This is a direct consequence of the assumption
$V_j(\varepsilon,\e^{\pm \I t}, y) 
\in \C[\varepsilon, \e^{\I t}, \e^{-\I t}, y_1, \ldots, y_n]$.
Therefore, replacing $t$ with $t - s$ in $P_m(\varepsilon, t, A)$  
yields another formal solution to (\ref{yeq}), as claimed.
\end{proof}

We refer to $A=(A_1,\ldots, A_n)$ as the {\em bare amplitudes}, and define the \emph{renormalized amplitudes}
to be the resonant secular coefficients. That is,
\begin{align}\label{ra3}
\mathscr{A}(\varepsilon,t,A)&=(\mathscr{A}_1(\varepsilon,t,A),\ldots,
\mathscr{A}_n(\varepsilon,t,A)),
\qquad
\mathscr{A}_j(\varepsilon,t,A) := P_{j,m_j}(\varepsilon,t,A).
\end{align}

The main result of this section is the following theorem and its consequences.
\begin{theorem}\label{th:yy3}
For any $s,t,$ and  $A=(A_1,\ldots, A_n)$, 
the following equality holds:
\begin{align}
\sum_{m \in Z}P_{j,m}(\varepsilon, t,A)\e^{\I mt}
= \sum_{m \in \Z} 
P_{j,m}(\varepsilon,t-s,\mathscr{A}(\varepsilon,s,A))\e^{\I mt}
\quad (1 \le j \le n).
\label{yy3}
\end{align}
\end{theorem}
\begin{proof}
It suffices to verify that the RHS of (\ref{yy3}) satisfies the conditions (i), (ii) and (iii)
in Lemma \ref{le:ren3}.
Condition (i) follows from Lemma \ref{le:ren2}.
Conditions (ii) and (iii) are shown as follows:
\begin{align*}
&\sum_{m \in \Z} 
P_{j,m}(0,t-s,\mathscr{A}(0,s,A))\e^{\I mt}
\overset{(\ref{pa3})}{=}
\sum_{m \in \Z} \mathscr{A}_j(0,s,A)\delta_{m,m_j}\e^{\I mt}
\overset{(\ref{ra3})}{=} P_{j,m_j}(0,s,A)\e^{\I m_jt} 
\overset{(\ref{pa3})}{=} A_j\e^{\I m_jt}, 
\\
&\left. \left[\sum_{m \in \Z} 
P_{j,m}(\varepsilon,t-s,\mathscr{A}(\varepsilon,s,A))\e^{\I mt}
\right]_{\e^{\I m_jt}} \right|_{t=s}
= P_{j,m_j}(\varepsilon,0,\mathscr{A}(\varepsilon,s,A))
\overset{(\ref{pa3})}{=}\mathscr{A}_j(\varepsilon,s,A)
\overset{(\ref{ra3})}{=} P_{j,m_j}(\varepsilon,s,A).
\end{align*} 
\end{proof}

\begin{corollary}\label{co:p=p3}
The secular coefficients satisfy the following functional relation:
\begin{align}\label{pp3}
P_{j,m}(\varepsilon, t,A)
= P_{j,m}(\varepsilon,t-s,\mathscr{A}(\varepsilon,s,A))
\qquad (1 \le j \le n, \,m \in \Z).
\end{align}
\end{corollary}

In particular, from the definition (\ref{ra3}), 
the case $m=m_j$ yields 
a closed functional equation
for the renormalized amplitudes:
\begin{align}\label{AA3}
\mathscr{A}_j(\varepsilon,t,A)
= \mathscr{A}_j(\varepsilon,t-s,\mathscr{A}(\varepsilon,s,A))
\qquad (1 \le j \le n).
\end{align}
Replacing $(t,s)$ with $(t+s,t)$ brings (\ref{AA3}) into 
$\mathscr{A}_j(\varepsilon,t+s,A)
= \mathscr{A}_j(\varepsilon,s,\mathscr{A}(\varepsilon,t,A))
= P_{j,m_j}(\varepsilon,s,\mathscr{A}(\varepsilon,t,A))$.
Differentiating this identity with respect to $s$ at $s=0$, we obtain the RG equation:
\begin{align}\label{rg3}
\frac{d}{dt}\mathscr{A}_j(\varepsilon,t,A) = 
\left. \frac{\partial}{\partial s}P_{j,m_j}(\varepsilon,s,\mathscr{A}(\varepsilon,t,A))\right|_{s=0}
\quad (1 \le j \le n).
\end{align}
This is a system of first-order autonomous differential equations for 
the renormalized amplitudes 
$\mathscr{A}_1(\varepsilon,t,A)$, $\ldots$, $\mathscr{A}_n(\varepsilon,t,A)$.

Suppose that the secular coefficients are obtained by a naive perturbation as 
\begin{align}\label{rg35}
P_{j,m}(\varepsilon,t,A) = \delta_{m,m_j}A_j + 
\sum_{k\ge 1, l, r_1,\ldots, r_n \ge 0}
B_{j,m;k,l,r_1, \ldots, r_n}
\varepsilon^k t^lA_1^{r_1}\cdots A_n^{r_n},
\end{align}
where the coefficient $B_{j,m;k,l, r_1,\ldots,r_n}$ vanishes when
$(m,l)=(m_j,0)$, reflecting \eqref{pa3}.
Then the RG equation (\ref{rg3})  is written as
\begin{equation}\label{rg4}
\frac{d{\mathscr A}_j}{dt} = 
\sum_{k\ge 1, r_1,\ldots, r_n \ge 0}
B_{j,m_j;k,1,r_1,\ldots, r_n}
\varepsilon^{k}\mathscr{A}_1^{r_1}\cdots \mathscr{A}_n^{r_n}
\qquad (1 \le j \le n),
\end{equation}
where $\mathscr{A}_j = \mathscr{A}_j(\varepsilon,t,A)$.
Note that the dynamics (\ref{rg4}) is {\em slow} in the sense that 
the RHS is at least of order $\mathcal{O}(\varepsilon)$.

By applying Theorem \ref{th:yy3} with $s=t$ and 
using the definition (\ref{ra3}),
the naive perturbative solution (\ref{Yv3}) can be rewritten 
in a form where all secular dependence on $t$ is absorbed 
into the renormalized amplitudes:
\begin{align}\label{rs3}
Y_j(\varepsilon,t,A) 
= \sum_{m \in \Z} 
P_{j,m}(\varepsilon,0,\mathscr{A}(\varepsilon,t,A))\e^{\I mt}
\quad (1 \le j \le n).
\end{align}
The renormalized expansion (\ref{rs3}), together with the RG equation 
(\ref{rg3}) completes our RG approach to the original equation (\ref{yeq}).

The definition (\ref{ra3}) for the renormalized amplitude 
$\mathscr{A}=\mathscr{A}(\varepsilon,t,A)$
in terms of the bare one $A$ can be inverted as 
\begin{align}\label{inv3}
A_j = A_j(\varepsilon,t,\mathscr{A})
=P_{j,m_j}(\varepsilon,-t,\mathscr{A})
\qquad (1 \le j \le n).
\end{align}
This identity follows from a combination of  (\ref{pa32}) and 
(\ref{pp3}) with $m=m_j$ and the pair $(s,t)$ replaced by $(t,0)$.
In this way, the resonant secular coefficients $P_{j,m_j}$ 
provide the {\em inversion relations} 
 (\ref{ra3}) and (\ref{inv3}), which enable one to express the 
bare amplitude from the renormalized one, and vice versa.

The functional relation \eqref{AA3}  for the renormalized amplitude \eqref{ra3} has the form
\begin{equation}
\mathscr{A}(t,A)=\mathscr{A}(t-s,\mathscr{A}(s,A)),
\end{equation}
where the dependence on $\varepsilon$ is suppressed.
This shows that $\{\mathscr{A}(t,\cdot)\}_{t\in\mathbb{R}}$ forms a
(formal) one-parameter group under composition. 
Hence there exists a
(formal) vector field $X$ on the amplitude space such that
$\mathscr{A}(t,A)=\exp(tX)\,A$.
The RHS of the RG equation~\eqref{rg4}  evaluated at $\mathscr{A}=A$ 
gives the coefficients of $\partial/\partial A_j$ in $X$.

\subsection{Autonomous case}\label{ss:auto}

Let us consider the autonomous case of the equation (\ref{yeq}), where 
$V$ takes the form $V(\varepsilon, y)$ not including $\e^{\I t}$ explicitly and depends on $t$ only through $y=y(t)$.
For any parameter $u$, define
\begin{align}\label{yh}
Y_u(\varepsilon,t,A) 
 = \sum_{m \in \Z}
 \begin{pmatrix}P_{1,m}(\varepsilon, t,\e^{-\I M u }A) \\ \vdots \\ P_{n,m}(\varepsilon,t,\e^{-\I M u}A) \end{pmatrix}
 \e^{\I m(t+u)},
 \end{align}
 where $M$ is the diagonal matrix in (\ref{ymf}), so that
 $\e^{-\I M u}A=(\e^{-\I m_1 u}A_1,\ldots, \e^{-\I m_n u}A_n)$.
From Lemma \ref{le:ren2}, we know that (\ref{yB})  with  $(t,s,B)$ replaced by $(t, u,\e^{-\I M u} A)$
gives a solution to (\ref{yeq}).
Moreover,  under the assumption that $V$ is autonomous, the solution remains valid even after shifting 
$t$ to $t+u$. 
Thus, $Y_u(\varepsilon,t,A)$ is indeed a solution to the autonomous case of (\ref{yeq}).
Now, we show that $Y_u(\varepsilon,t,A) $ furthermore satisfies the property (\ref{shokou}) at $\varepsilon=0$ 
and the behavior (\ref{ycon3}).
Both can be easily verified using (\ref{pa3}) as follows:
\begin{align}
Y_u(0,t,A)&= \sum_{m \in \Z}
 \begin{pmatrix}P_{1,m}(0, t,\e^{-\I M u }A) 
 \\ \vdots \\ 
 P_{n,m}(0,t,\e^{-\I M u}A) \end{pmatrix}
 \e^{\I m(t+u)}
=
 \begin{pmatrix} 
 \e^{-\I m_1 u}A_1 \e^{\I m_1(t+u)}
 \\ \vdots \\ 
\e^{-\I m_n u}A_n \e^{\I m_n(t+u)}
\end{pmatrix},
\\
\label{yh2}
\begin{split}
Y_u(\varepsilon,t,A) 
& = \sum_{m \in \Z}
 \begin{pmatrix}
\delta_{m,m_1}\e^{-\I m_1 u}A_1+\mathcal{O}(\varepsilon t^{\delta_{m,m_1}})
 \\ \vdots \\
 \delta_{m,m_n}\e^{-\I m_n u}A_n+\mathcal{O}(\varepsilon t^{\delta_{m,m_n}})
 \end{pmatrix}
 \e^{\I m(t+u)}
 \\
 & =  \begin{pmatrix}
 \bigl((A_1+\e^{\I m_1 u}\mathcal{O}(\varepsilon t)\bigr) \e^{\I m_1t}
 + \sum_{m \neq m_1}\e^{\I m u}\mathcal{O}(\varepsilon)\e^{i m t}
 \\ \vdots \\
 \bigl((A_n+\e^{\I m_n u}\mathcal{O}(\varepsilon t)\bigr) \e^{\I m_nt}
 + \sum_{m \neq m_n}\e^{\I m u}\mathcal{O}(\varepsilon)\e^{i m t}
 \end{pmatrix}.
 \end{split}
\end{align}
The behavior described in (\ref{yh2}) is equivalent to (\ref{ycon3}).
Therefore, by the uniqueness of the solution discussed between (\ref{shokou}) and (\ref{ycon3}), we conclude that
\begin{align}\label{y=yh}
Y_u(\varepsilon,t,A) = Y(\varepsilon,t,A) \; \;\text{for autonomous $V$},
\end{align}
for any $u$. 
As a corollary, comparing (\ref{yh}) and (\ref{Yv3}), we find that 
the secular coefficients satisfy
\begin{align}\label{phom}
P_{j,m}(\varepsilon,t,A)  = \e^{\I m u} P_{j,m}(\varepsilon,t, \e^{-\I M u}A).
\end{align}
The property \eqref{phom} implies that the coefficients $B_{j,m;k,l,r_1,\ldots, r_n}$ in 
\eqref{rg35} vanish unless $r_1m_1+\cdots + r_nm_n=m$ is satisfied.
As a consequence, the resulting RG equation (\ref{rg4}) is invariant under the transformation 
$\mathscr{A}_j \mapsto z^{m_j}\mathscr{A}_j$ for any $z$.
Such a homogeneous system 
is regarded as a {\em normal form} of the original equation. 
See \cite{C2, EFK, M} and references therein.

In the further special case $M=0$, 
one readily verifies that $P_{j,m}(\varepsilon,t,A)=0$ for $m\neq0$, which implies
\begin{align}\label{ypa}
Y_j(\varepsilon,t,A)=P_{j,0}(\varepsilon,t,A)=\mathscr{A}_j(\varepsilon,t,A).
\end{align}
Hence the RG equation reduces to the original system
$d\mathscr{A}_j/dt=\varepsilon\,V_j(\varepsilon,\mathscr{A})$.

\subsection{Examples}\label{ss:exss}

In this subsection, we present two explicit examples within the general framework.
They illustrate the structure of the secular coefficients and the RG equation in concrete settings,
and lead to nonlinear amplitude-phase equations.

\begin{example} (A non-autonomous two-dimensional system)\label{ex:CD}
Consider a two-dimensional non-autonomous system:
\begin{align}\label{ex1:eq}
\frac{d}{dt}
\begin{pmatrix} y_1 \\ y_2 \end{pmatrix}
=
\begin{pmatrix} \I y_1 \\  -\I y_2 \end{pmatrix}
+ \varepsilon 
\begin{pmatrix} y_1 y_2+ \e^{-\I t} y_2  \\  y_1y_2 + \e^{\I t}y_1 \end{pmatrix},
\end{align}
which is an example of  \eqref{yeq}--\eqref{ymf} with $M=\mathrm{diag}(1,-1)$.
The secular coefficients $P_{1,m}= P_{1,m}(\varepsilon, t, (A_1, A_2))$ 
for $|m|\le 3$ up to $\mathcal{O}(\varepsilon^5)$ read
\begin{align*}
P_{1,-3}&= -\frac{A_2^2 \varepsilon^2}{12}+
\frac{1}{432} A_2^2 \varepsilon^4 \left(18 A_1 A_2^2-72 \I A_1 A_2 t+12 A_1 A_2-72 A_1-24 \I t-1\right),
\\
P_{1,-2}&=\frac{\I A_2 \varepsilon}{3}-\frac{1}{54} \I A_2 \varepsilon^3 \left(9 A_1 A_2^2-18 \I A_1 A_2 t-18 A_1 A_2-6 \I t-2\right),
\\
P_{1,-1}&=-\frac{1}{2} A_1 (A_2-1) A_2 \varepsilon^2
+\frac{1}{36} A_1 A_2 \varepsilon^4 \left(-18 \I A_1 A_2^2 t-27 A_1 A_2^2+27 A_1 A_2-6 \I A_2 t-25 A_2+11\right),
\\
P_{1,0}&= \I A_1 A_2 \varepsilon+\frac{1}{18} \I A_1 A_2 \varepsilon^3 (18 A_1 A_2-9 A_1+11),
\\
P_{1,1}&= A_1-\frac{1}{3} \I A_1 \varepsilon^2 t (3 A_1 A_2+1)
\\
&-\frac{1}{54} A_1 \varepsilon^4 t\left(\I(54 A_1^2 A_2^2-9 A_1^2 A_2-27 A_1 A_2^2+57 A_1 A_2+2)+3 t (3 A_1 A_2+1)^2\right),
\\
P_{1,2}&= \frac{1}{12} \I A_1^2 \varepsilon^3 (6 A_1 A_2-8 A_2+1),
\\
P_{1,3}&= -\frac{A_1^2 \varepsilon^2}{6}
+\frac{A_1^2 \varepsilon^4 \left(90 A_1^2 A_2+360 \I A_1 A_2 t-300 A_1 A_2+9 A_1+120 \I t-62\right)}{1080}.
\end{align*}
The other series of the secular coefficients are given by 
$P_{2,m}(\varepsilon,t, (A_1,A_2)) =  \left.P_{1,-m}(\varepsilon,t, (A_2,A_1))\right|_{\I \rightarrow -\I}$.
The functional relation \eqref{pp3} can be checked in small orders of $\varepsilon$.

The renormalized amplitudes are defined by 
$\mathscr{A}_1 = P_{1,1}$ and $\mathscr{A}_2= P_{2,-1}$.
According to \eqref{rg3}, their dynamics is governed by the RG equation:
\begin{equation}\label{ex1:AA}
\begin{split}
\frac{d{\mathscr A}_1}{dt} 
&= -\frac{\I \varepsilon^2}{3} \mathscr{A}_1\left(1+3 \mathscr{A}_1 \mathscr{A}_2\right)
-\frac{\I \varepsilon^4}{54}\mathscr{A}_1(2+57 \mathscr{A}_1 \mathscr{A}_2
-27 \mathscr{A}_1 \mathscr{A}_2^2-9 \mathscr{A}_1^2 \mathscr{A}_2
+54 \mathscr{A}_1^2 \mathscr{A}_2^2)+\mathcal{O}(\varepsilon^6),
\\
\frac{d{\mathscr A}_2}{dt} 
&= \frac{\I \varepsilon^2}{3} \mathscr{A}_2\left(1+3 \mathscr{A}_1 \mathscr{A}_2\right)
+\frac{\I \varepsilon^4}{54}\mathscr{A}_2(2+57 \mathscr{A}_1 \mathscr{A}_2
-9 \mathscr{A}_1 \mathscr{A}_2^2-27 \mathscr{A}_1^2 \mathscr{A}_2
+54 \mathscr{A}_1^2 \mathscr{A}_2^2)+\mathcal{O}(\varepsilon^6).
\end{split}
\end{equation}
The renormalized expansion \eqref{rs3} is given by 
\begin{equation}\label{ex1:yy}
\begin{split}
Y_1 &=\mathscr{A}_1 \e^{i t}
+ \frac{\varepsilon}{3} \left(
 3 i \mathscr{A}_1 \mathscr{A}_2
 + i \mathscr{A}_2 \e^{-2 i t}
\right)
+ \frac{\varepsilon^2}{12} \left(
 6 \mathscr{A}_1 \mathscr{A}_2 \e^{-i t}
 - 6 \mathscr{A}_1 \mathscr{A}_2^{2} \e^{-i t}
 - \mathscr{A}_2^{2} \e^{-3 i t}
 - 2 \mathscr{A}_1^{2} \e^{3 i t}
\right)+ \mathcal{O}(\varepsilon^3),
\\
Y_2 &=\mathscr{A}_2 \e^{-i t}
- \frac{\varepsilon}{3} \left(
 3 i \mathscr{A}_1 \mathscr{A}_2
 + i \mathscr{A}_1 \e^{2 i t}
\right)
+ \frac{\varepsilon^2}{12} \left(
 6 \mathscr{A}_1 \mathscr{A}_2 \e^{i t}
 - 6 \mathscr{A}_1^{2} \mathscr{A}_2 \e^{i t}
 - 2 \mathscr{A}_2^{2} \e^{-3 i t}
 - \mathscr{A}_1^{2} \e^{3 i t}
\right)+ \mathcal{O}(\varepsilon^3).
\end{split}
\end{equation}
In both \eqref{ex1:AA} and \eqref{ex1:yy}, the two equations are transformed to each other by the interchange 
${\mathscr A}_1 \leftrightarrow {\mathscr A}_2$ 
with $\I \rightarrow -\I$.

The equation \eqref{ex1:eq} for complex $y_1$ and $y_2$ can be consistently restricted to the case $y_2 = y_1^*$.
It is then natural to introduce the renormalized magnitude $R$ and phase $\theta$ by setting
$\mathscr{A}_1 = R \e^{\I \theta}$ and $\mathscr{A}_2 = R \e^{-\I \theta}$.
The RG equation \eqref{ex1:AA} can then be rewritten as follows:
\begin{equation}\label{rteq}
\begin{split}
\frac{d R}{dt} 
&=\frac{1}{3} \varepsilon^{4} R^{4} \sin \theta
+ \frac{1}{180} \varepsilon^{6} R^{4} \left(97 + 175 R^{2}\right) \sin \theta +\mathcal{O}(\varepsilon^8),
\\
\frac{d\theta}{dt}
&=-\frac{1}{3} \varepsilon^{2} \left(1 + 3 R^{2}\right)
+ \frac{1}{54} \varepsilon^{4} \left(-2 - 57 R^{2} - 54 R^{4} + 36 R^{3} \cos \theta \right) 
\\
&+ \frac{\varepsilon^{6}}{9720} \left(
-80 - 7023 R^{2} - 32913 R^{4} - 21870 R^{6}
+ 11610 R^{3} \cos \theta + 28890 R^{5} \cos \theta
\right) +\mathcal{O}(\varepsilon^8).
\end{split}
\end{equation}
See Appendix \ref{app:n} for representative plots obtained by numerical integration of these equations.
\end{example}

\begin{example}(Coupled oscillators) \label{ex:mar23}
Consider a system of $n$-coupled oscillators governed by the Hamilton equations of motion
\begin{align}\label{co}
\dot{q}_j = p_j, \qquad 
\dot{p}_j = -m_j^2 q_j + \varepsilon V_j(\varepsilon, \e^{\pm \I t}, q_1, \ldots, q_n)
\qquad (1 \le j \le n),
\end{align}
where $V_j(\varepsilon, \e^{\pm \I t}, q_1,\ldots, q_n)$ is a polynomial in 
$\varepsilon$ and in the canonical variables $q_1, \ldots, q_n$.
It may also be a (real) Laurent polynomial in $\e^{\I t}$, corresponding to a periodic external forcing.
We assume a commensurate situation such that $m_1,\ldots, m_n \in \Z_{\ge 1}$.
Then, switching to the variables $y_1,\ldots, y_{2n}$ defined by 
$y_{2j-1} = p_j + \I m_j q_j$ and $y_{2j} = p_j - \I m_j q_j$,
the system \eqref{co} is transformed into the form \eqref{yeq} with $n$ replaced by $2n$
and $M$ taking the form $\mathrm{diag}(m_1,-m_1,\ldots, m_n, -m_n)$.

As a concrete example, we consider the coupled nonlinear oscillators 
\begin{equation}\label{q12}
\ddot{q}_1 + q_1 + 4 \varepsilon q_2 \dot{q}_1 = 0,
\qquad
\ddot{q}_2 + q_2 + 4 \varepsilon q_1 \dot{q}_2 = 0.
\end{equation}
Applying the above formulation, the system \eqref{q12} is brought to the form of \eqref{yeq} as
\begin{subequations}
\begin{align}
&\frac{d}{dt}\begin{pmatrix}y_1 \\ y_2 \\ y_3 \\ y_4 \end{pmatrix}
= \I \begin{pmatrix}1 & 0 & 0 & 0 \\ 0 & -1 & 0 & 0 \\ 0 & 0 & 1 & 0 \\ 0 & 0 & 0 & -1 \end{pmatrix}
\begin{pmatrix}y_1 \\ y_2 \\ y_3 \\ y_4 \end{pmatrix} + 
\varepsilon \begin{pmatrix} V_1 \\ V_2 \\  V_3 \\ V_4 \end{pmatrix},
\\
&
V_1 = V_2  = \I (y_1+y_2) (y_3-y_4),\qquad
V_3 = V_4  = \I (y_1-y_2) (y_3+y_4) .
\end{align}
\end{subequations}
The renormalized amplitudes \eqref{ra3} are given by 
\begin{equation}
\mathscr{A} = (\mathscr{A}_1,\mathscr{A}_2,\mathscr{A}_3,\mathscr{A}_4)
= (P_{1,1}(\varepsilon,t,A), P_{2,-1}(\varepsilon,t,A), P_{3,1}(\varepsilon,t,A), P_{4,-1}(\varepsilon,t,A)),
\end{equation}
where $A=(A_1,A_2,A_3,A_4)$ denotes the bare amplitudes.
According to \eqref{shokou}, 
naive perturbation starting from 
$(y^{(0)}_1,y^{(0)}_2,y^{(0)}_3,y^{(0)}_4) 
=(A_1\e^{\I t}, A_2 \e^{-\I t}, A_3\e^{\I t}, A_4\e^{-\I t})$
yields
\begin{align*}
P_{1,1}(\varepsilon,t,A) &=A_1
- \frac{2}{3} \mathrm{i} A_1 \left(-4 A_2 A_3 + 3 A_1 A_4 + 2 A_3 A_4\right) t \varepsilon^2
+\mathcal{O}(\varepsilon^4),
\\
P_{2,-1}(\varepsilon,t,A) 
&=\left. P_{1,1}(\varepsilon,t,A) \right|_{A_1\leftrightarrow A_2, A_3\leftrightarrow A_4, \I \rightarrow -\I},
\\
P_{3,1}(\varepsilon,t,A) 
&=\left. P_{1,1}(\varepsilon,t,A) \right|_{A_1\leftrightarrow A_3, A_2\leftrightarrow A_4, h_1 \leftrightarrow h_2},
\\
P_{4,-1}(\varepsilon,t,A) 
&=\left. P_{1,1}(\varepsilon,t,A) \right|_{A_1\leftrightarrow A_4, 
A_2\leftrightarrow A_3, h_1 \leftrightarrow h_2,\I \rightarrow -\I}.
\end{align*}
The RG equation \eqref{rg3} reads
\begin{equation}\label{rg3.5}
\frac{d}{dt} (\mathscr{A}_1,\mathscr{A}_2,\mathscr{A}_3,\mathscr{A}_4)
= \left. \frac{\partial }{\partial s}
(P_{1,1}(\varepsilon,s,\mathscr{A}), 
P_{2,-1}(\varepsilon,s,\mathscr{A}), 
P_{3,1}(\varepsilon,s,\mathscr{A}), 
P_{4,-1}(\varepsilon,s,\mathscr{A}))\right|_{s=0}.
\end{equation}
Introduce the magnitude and phases of the renormalized amplitudes as
\begin{align}\label{ART}
{\mathscr A}_1 = R_1\e^{\I \theta_1}, \quad
{\mathscr A}_2 = R_1\e^{-\I \theta_1}, \quad
{\mathscr A}_3 = R_2\e^{\I \theta_2}, \quad
{\mathscr A}_4 = R_2\e^{-\I \theta_2}.
\end{align}
Then the RG equation \eqref{rg3.5} becomes 
\begin{subequations}
\begin{align}
\frac{dR_1}{dt}
&= \frac{14}{3} \varepsilon^{2} R_1^{2} R_2 \sin \theta_{12}
\nonumber \\
&- \frac{1}{27} \varepsilon^{4} R_1 R_2 \Bigl[
(70 R_1^{3} + 274 R_1 R_2^{2}) \sin \theta_{12}
+ (43 R_1^{2} R_2 + 72 R_2^{3}) \sin(2\theta_{12})
\Bigr]+ \mathcal{O}(\varepsilon^6),
\label{Rt1a}
\\[1ex]
\frac{d\theta_1}{dt}
&= -\frac{2}{3} \varepsilon^{2} R_2 \bigl(2 R_2 - R_1 \cos \theta_{12}\bigr)
\nonumber\\
&+ \frac{1}{27} \varepsilon^{4} R_2 \Bigl[
(-50 R_1^{2} R_2 - 52 R_2^{3})
+ (2 R_1^{3} + 90 R_1 R_2^{2}) \cos \theta_{12}
+ (65 R_1^{2} R_2 - 72 R_2^{3}) \cos(2\theta_{12})
\Bigr]+ \mathcal{O}(\varepsilon^6).
\label{Rt1b}
\end{align}
\end{subequations}
The equation for $dR_2/dt$ and $d\theta_2/dt$ are obtained 
by interchanging $(R_1, \theta_1)$ and $(R_2, \theta_2)$ in these equations.

Substituting \eqref{ART} into $Y_j$ in \eqref{rs3}, we obtain the renormalized expansion 
of $q_1 = (Y_1-Y_2)/(2\I)$ as follows:
\begin{equation}
\begin{split}
q_1
&= R_1 \sin(t+\theta_1)
+ \frac{2}{3} \varepsilon R_1 R_2 \Bigl(
3 \sin(\theta_1-\theta_2)
+ \sin(2t+\theta_1+\theta_2)
\Bigr)
\\
&+ \frac{1}{6} \varepsilon^{2} R_1 R_2 \Bigl[
4 R_2 \sin(t+\theta_1)
+ 6 R_1 \sin(t+2\theta_1-\theta_2)
- 8 R_1 \sin(t+\theta_2)
\\
&
\qquad \qquad \qquad + R_1 \sin(3t+2\theta_1+\theta_2)
+ 2 R_2 \sin(3t+\theta_1+2\theta_2)
\Bigr]
+ \mathcal{O}(\varepsilon^4).
\end{split}
\end{equation}
The result for $q_2$ follows form this 
by interchanging $(R_1, \theta_1)$ and $(R_2, \theta_2)$.
\end{example}

\medskip

The system \eqref{eq1}  also accommodates other interesting classes of models,
including certain Lotka--Volterra type equations under suitable conditions on the coefficients.

\section{First-Order ODE System with a Nilpotent Coefficient Matrix}\label{sec:5}

We now consider the first-order ODE (\ref{eq1}) 
for the $n$-dimensional vector $y$, 
with the matrix $\I M$ taken to be a single Jordan block with 
eigenvalue $\I m$ with $m \in \Z$.
That is, we consider the system
\begin{subequations}\label{eq2}
\begin{align}
\frac{dy}{dt} &= \I M y + \varepsilon V(\varepsilon,\e^{\pm \I t},y),
\label{yeq4}\\
\I M= & \I m\, \mathrm{Id} + \Lambda,
\qquad
\Lambda = 
\begin{pmatrix}
 0 & \!1 \\
  & \ddots & \ddots \\
  &        & 0 & 1 \\
   &        &   & 0
\end{pmatrix}.
\label{M4}
\end{align}
Here, $\mathrm{Id}$ and $\Lambda$ are $n$-dimensional matrices
given by $\mathrm{Id} = (\delta_{i,j})_{1 \le i,j \le n}$, and 
$\Lambda=(\delta_{i+1,j})_{1 \le i,j \le n}$.
As in \eqref{eq1}, we assume that 
$V_j(\varepsilon,\e^{\pm \I t}, y) \in 
\C[\varepsilon, \e^{\I t}, \e^{-\I t}, y_1, \ldots, y_n]$ for each $1 \le j \le n$.
Introducing the transformation $y = \e^{\I mt}\tilde{y}$ reduces (\ref{yeq4}) to 
$d\tilde{y}/dt =  \Lambda \tilde{y} + \varepsilon
\e^{-\I m t}V(\varepsilon,\e^{\pm \I t},\e^{\I m t}{\tilde y})$.
From $m \in \Z$,  we see that 
$\e^{-\I m t}V_j(\varepsilon,\e^{\pm \I t},\e^{\I m t}{\tilde y})$ remains 
a polynomial in $\varepsilon, {\tilde y}_1,\ldots, {\tilde y}_n$ and a Laurent polynomial
in $\e^{\I t}$.
Hence, without loss of generality, we may set $m=0$ 
and study the simplified system
\begin{align}\label{yeq42}
\frac{dy}{dt} &= \Lambda y + \varepsilon V(\varepsilon,\e^{\pm \I t},y),
\end{align}
\end{subequations}
where  $y= y(\varepsilon,t)=(y_1(\varepsilon,t),\ldots, y_n(\varepsilon,t))$ and 
$V(\varepsilon,\e^{\pm \I t},y) =
(V_1(\varepsilon,\e^{\pm \I t},y),\ldots, V_n(\varepsilon,\e^{\pm \I t},y))$ 
with the components 
$V_j(\varepsilon,\e^{\pm \I t},y) \in \C[\varepsilon, \e^{\I t}, \e^{-\I t}, y_1, \ldots, y_n]$.
We do not assume $V(\varepsilon, \e^{\pm \I t}, y=0)=0$.
The title of this section refers to the fact that the matrix $\Lambda$ (\ref{M4}) 
appearing in 
 (\ref{yeq42}) is nilpotent, in contrast to the semisimple 
 case considered previously in (\ref{ymf}).
A further change of variables $y \to \hat y = \e^{-\Lambda t} y$ will not be used,
because $\e^{\pm\Lambda t}$ is a matrix with polynomial entries, and therefore
$\e^{-\Lambda t} V(\varepsilon, \e^{\pm \I t}, \e^{\Lambda t} \hat y)$ does not
satisfy the condition stated above.

\subsection{Naive perturbation}
Substituting the series expansion (\ref{eexp}) in $\varepsilon$ into (\ref{yeq42}), 
we obtain a sequence of equations at each order of $\varepsilon$:
\begin{subequations}
\begin{align}
\frac{dy^{(0)}}{dt} &= \Lambda y^{(0)},
\label{yeq02}
\\
\frac{dy^{(1)}}{dt} &= \Lambda y^{(1)}+ V(0,\e^{\pm \I t}, y^{(0)}),
\nonumber \\
 & \vdots
\nonumber\\
\frac{dy^{(k)}}{dt} &= \Lambda y^{(k)}+ \Biggl[
V(\varepsilon, \e^{\pm \I t}, \sum_{0 \le j <  k}\varepsilon^j y^{(j)})
\Biggr]_{\varepsilon^{k-1}},
\label{yeqk2}
\end{align}
\end{subequations}
where the symbol $[X ]_{\alpha}$ has been defined after (\ref{yeqk}).
The general solution to the first equation (\ref{yeq02}) is given by 
$y^{(0)} = \e^{\Lambda t}A$ for a constant vector $A$, i.e., 
\begin{subequations}\label{yg0}
 \begin{align}\label{shokou2}
 \begin{pmatrix}y_1^{(0)}(t) \\ y_2^{(0)}(t) \\ \vdots  \\ y_n^{(0)}(t)  \end{pmatrix}
 &=\begin{pmatrix}
 g(t,A)   \\  \frac{dg(t,A)}{dt} \\ \vdots  \\  \frac{d^{n-1}g(t,A)}{dt^{n-1}} \end{pmatrix},
 \\
 g(t,A) &= \sum_{k=1}^{n}\frac{A_kt^{k-1}}{(k-1)!},
 \label{gdef}
 \end{align}
 \end{subequations}
 where $A=(A_1,\ldots, A_n)$ is an $n$-tuple of arbitrary parameters.
 For instance for $n=3$, we have
 \begin{align}\label{ye3}
 y^{(0)}(t) = \begin{pmatrix}
 A_1+ A_2t + \frac{A_3t^2}{2}\\
 A_2+A_3 t\\
 A_3
 \end{pmatrix}.
 \end{align}
Higher-order terms with respect to $\varepsilon$ are uniquely determined by 
imposing the condition
 \begin{align}\label{yjk4}
\left.  \left[ y^{(k)}_j(\varepsilon,t)\right]_{\e^{\I0 t}}\right|_{t=0}=0 
\quad (1 \le j \le n, \, k \ge 1),
\end{align}
In the notation of (\ref{yf3}), we have 
$\left[ y^{(k)}_j(\varepsilon,t)\right]_{\e^{\I0 t}} = f^{(k)}_{j,0}(t)$.
Thus, condition (\ref{yjk4})  implies that the general solution to the 
homogeneous equation at each order $k \ge 1$ must be fixed so that
$f^{(k)}_{j,0}(t=0)=0$ 
for all $1 \le j \le n$.

Let $Y(\varepsilon, t, A)=(Y_1(\varepsilon, t, A), \ldots, Y_n(\varepsilon, t, A))$ 
be the resulting solution, and 
define the associated {\em secular coefficients} $P_{j,m}(\varepsilon,t,A)$ 
in the same manner as (\ref{Yv3}):
 \begin{align}\label{Yv4}
 \begin{pmatrix}Y_1(\varepsilon,t,A) \\ \vdots \\ Y_n(\varepsilon,t,A) \end{pmatrix}
 = \sum_{m \in \Z}
 \begin{pmatrix}P_{1,m}(\varepsilon, t,A) \\ \vdots \\ P_{n,m}(\varepsilon,t,A) \end{pmatrix}
 \e^{\I mt}.
 \end{align}
The special case $P_{j,0}(\varepsilon,t,A)$ will be referred to as 
{\em resonant secular coefficients}.
By definition, $P_{j,m}(\varepsilon,t,A)$ satisfies 
\begin{align}\label{pa4}
P_{j,m}(\varepsilon,t,A) = \frac{d^{j-1}g(t,A)}{dt^{j-1}} \delta_{m,0} + \od(\varepsilon t^{\delta_{m,0}})
\quad 
\text{as a formal power series in $\varepsilon$ and $t$}.
\end{align}
From (\ref{gdef}), the special case $m=0$ and $t=0$ in this relation yields
\begin{align}\label{pa42}
P_{j,0}(\varepsilon,0,A) = \left. \frac{d^{j-1}g(t,A)}{dt^{j-1}}\right|_{t=0} = A_j.
\end{align}
We will use the property \eqref{ped3}, which remains valid in the present setting.

\subsection{Functional equation for secular coefficients}

Now we state an analogue of Lemma \ref{le:ren3} in the nilpotent setting.

\begin{lemma}\label{le:ren4}
Let $s$ be an arbitrary parameter and $A=(A_1,\ldots, A_n)$.
Then the formal power series $y=y(\varepsilon,t)$ of the form (\ref{yf3}) 
that satisfies the conditions (i), (ii), and (iii) below is unique, 
and coincides with $Y(\varepsilon,t,A)$ as defined in (\ref{Yv4}).
\begin{align*}
&(\I)\;  y(\varepsilon,t) \, \text{satisfies the differential equation (\ref{yeq42})},
\\
&(\I\I)\;  y(0,t)=\bigl(g(t,A),  \frac{dg(t,A)}{dt}, \ldots, \frac{d^{n-1}g(t,A)}{dt^{n-1}} \bigr),
\\
&(\I\I\I)\;  
[y_j(\varepsilon,t)]_{\e^{\I 0 t}}|_{t=s} = P_{j,0}(\varepsilon,s,A)\quad   (1 \le j \le n).
\end{align*}
\end{lemma}

Conditions (ii) and (iii) in Lemma \ref{le:ren4} 
differ slightly from those in Lemma \ref{le:ren3}. 
The following Lemma is an analogue of Lemma \ref{le:ren2}, 
and can be proved in the same manner.

\begin{lemma}\label{le:ren5}
Let $s$ be an arbitrary parameter, and let $B=(B_1,\ldots, B_n)$.
Then the formal power series
\begin{align}
\sum_{m\in \Z}
 \begin{pmatrix}P_{1,m}(\varepsilon, t-s,B) \\ \vdots \\ P_{n,m}(\varepsilon,t-s,B) \end{pmatrix}
 \e^{\I mt}
 \end{align}
 is also a solution to the equation (\ref{yeq4}).
\end{lemma}

We refer to $n$-tuple $A=(A_1,\ldots, A_n)$  
as the {\em bare amplitudes}, and define the {\em renormalized amplitudes} by 
\begin{align}\label{ra4}
\mathscr{A}(\varepsilon,t,A)&=(\mathscr{A}_1(\varepsilon,t,A),\ldots,
\mathscr{A}_n(\varepsilon,t,A)),
\qquad
\mathscr{A}_j(\varepsilon,t,A) := P_{j,0}(\varepsilon,t,A).
\end{align}
The main theorem of this subsection is the following, which 
formally takes the identical form as Theorem \ref{th:yy3}.

\begin{theorem}\label{th:yy4}
For any $s,t,$ and  $A=(A_1,\ldots, A_n)$, 
the following equality holds:
\begin{align}
\sum_{m \in Z}P_{j,m}(\varepsilon, t,A)\e^{\I mt}
= \sum_{m \in \Z} 
P_{j,m}(\varepsilon,t-s,\mathscr{A}(\varepsilon,s,A))\e^{\I mt}
\quad (1 \le j \le n).
\label{yy4}
\end{align}
\end{theorem}
\begin{proof}
It suffices to verify that the RHS of (\ref{yy4}) satisfies the conditions (i), (ii) and (iii)
in Lemma \ref{le:ren4}.
Condition (i) follows from Lemma \ref{le:ren5}.
For the $j$th component,  conditions (ii) is verified as follows:
\begin{equation*}
\begin{split}
&\sum_{m \in \Z} 
P_{j,m}(0,t-s,\mathscr{A}(0,s,A))\e^{\I mt}
\\
&\overset{(\ref{pa4})}{=}
\frac{d^{j-1}g(t-s,\mathscr{A}(0,s,A))}{d(t-s)^{j-1}}
=\left. \frac{d^{j-1}g(t,\mathscr{A}(0,s,A))}{dt^{j-1}}\right|_{t\rightarrow t-s}
\overset{(\ref{gdef})}{=} \left. 
\frac{d^{j-1}}{dt^{j-1}}\sum_{k=1}^n\frac{\mathscr{A}_k(0,s,A)t^{k-1}}{(k-1)!}\right|_{t\rightarrow t-s}
\\
&\overset{(\ref{ra4})}{=}
\frac{d^{j-1}}{dt^{j-1}}\sum_{k=1}^n \left.\frac{P_{k,0}(0,s,A)t^{k-1}}{(k-1)!}\right|_{t\rightarrow t-s}
\overset{(\ref{pa4})}{=}
\frac{d^{j-1}}{dt^{j-1}}\sum_{k=1}^n 
\left.\frac{t^{k-1}}{(k-1)!}\frac{d^{k-1}g(s,A)}{ds^{k-1}}\right|_{t\rightarrow t-s}
\\
&= \left. \frac{d^{j-1}g(t+s,A)}{dt^{j-1}}\right|_{t\rightarrow t-s} 
= \frac{d^{j-1}g(t,A)}{dt^{j-1}}.
\end{split}
\end{equation*}
In the first equality of the last line, we have used the fact that $g(t, A)$, as defined in (\ref{gdef}), 
is a polynomial in $t$ of degree at most $n - 1$.
Condition (iii) is shown as follows:
\begin{equation*}
\left. \left[ \sum_{m \in \Z} 
P_{j,m}(\varepsilon,t-s,\mathscr{A}(\varepsilon,s,A))\e^{\I mt}
\right]_{\e^{\I 0 t}}\right|_{t=s}
= P_{j,0}(\varepsilon,0,\mathscr{A}(\varepsilon,s,A))
\overset{(\ref{pa42})}{=}\mathscr{A}_j(\varepsilon,s,A)
\overset{(\ref{ra4})}{=}P_{j,0}(\varepsilon,s,A).
\end{equation*}
\end{proof}

As in the semisimple case, Theorem~\ref{th:yy4} yields a number of results
concerning the renormalized amplitudes, as follows.
First, we have the functional equation for the secular coefficients:
\begin{align}\label{pp32}
P_{j,m}(\varepsilon, t,A)
= P_{j,m}(\varepsilon,t-s,\mathscr{A}(\varepsilon,s,A))
\qquad (1 \le j \le n, \,m \in \Z).
\end{align}
This leads to the renormalized expansion,  which has a form identical to (\ref{rs3}).

Second, by the definition (\ref{ra4}) of $\mathscr{A}(\varepsilon,t,A)$, 
this includes the closed functional equation for the renormalized amplitudes:
\begin{align}\label{AA32}
\mathscr{A}_j(\varepsilon,t,A)
= \mathscr{A}_j(\varepsilon,t-s,\mathscr{A}(\varepsilon,s,A))
\qquad (1 \le j \le n).
\end{align}

Third, the RG equation for renormalized amplitudes takes the form
\begin{align}\label{rg32}
\frac{d}{dt}\mathscr{A}_j(\varepsilon,t,A) = 
\left. \frac{\partial}{\partial s}P_{j,0}(\varepsilon,s,\mathscr{A}(\varepsilon,t,A))\right|_{s=0}
\quad (1 \le j \le n).
\end{align}

Fourth, the inversion relation for (\ref{ra4}),
which expresses the bare amplitude $A=(A_1,\ldots, A_n)$ as a function 
of the renormalized amplitude $\mathscr{A}=\mathscr{A}(\varepsilon,t,A)$, 
is given by 
\begin{align}\label{inv4}
A_j = P_{j,0}(\varepsilon,-t,\mathscr{A})
\qquad (1 \le j \le n).
\end{align}
The relations \eqref{AA32}--\eqref{inv4} 
formally correspond to setting $m_j = 0$ in the analogous expressions 
\eqref{AA3}, \eqref{rg3} and \eqref{inv3} obtained in the semisimple case.

Suppose that the secular coefficients are obtained by a naive perturbation as
\begin{align}
P_{j,m}(\varepsilon,A) = \frac{d^{j-1}g(t,A)}{dt^{j-1}} \delta_{m,0}
+ \sum_{k\ge 1,l, r_1,\ldots, r_n \ge 0}
C_{j,m;k,l,r_1, \ldots, r_n}
\varepsilon^k t^lA_1^{r_1}\cdots A_n^{r_n},
\end{align}
where the coefficient $C_{j,m;k,l,r_1, \ldots, r_n}$ vanishes when 
$(m,l)=(0,0)$, reflecting \eqref{pa4}.
The RG equation \eqref{rg32} then takes the form
\begin{align}\label{rg5}
\frac{d\mathscr{A}_j}{dt}=  \mathscr{A}_{j+1}  + \sum_{k\ge 1,  r_1,\ldots, r_n \ge 0}
C_{j,0;k,1,r_1, \ldots, r_n}
\varepsilon^k \mathscr{A}_1^{r_1}\cdots \mathscr{A}_n^{r_n}
\qquad (1 \le j \le n),
\end{align}
where $\mathscr{A}_j = \mathscr{A}_j(\varepsilon,t,A)$ for $1 \le j \le n$ and 
$\mathscr{A}_{n+1}=0$.

In contrast to the semisimple case \eqref{rg4}, the dynamics in \eqref{rg5}
is {\em not} slow in general due to the presence of the first term on the RHS.
In the nilpotent case, all eigenvalues of the linear part are zero and
the Jordan blocks are nontrivial. 
Consequently there is no slow--fast separation, 
and the RG vector field emerges at order $\mathcal{O}(\varepsilon^0)$.

When $V$ is autonomous, 
i.e. of the form $(V_j(\varepsilon, y))_{j=1}^n$, the secular coefficients $P_{j,m}(\varepsilon, t, A)$ are
vanishing for $m \neq 0$.
Therefore the relation \eqref{ypa} holds also in the nilpotent setting, and 
the RG equation reduces to the original system
$d\mathscr{A}_j/dt=\mathscr{A}_{j+1}+ \varepsilon\,V_j(\varepsilon,\mathscr{A})$
with $\mathscr{A}_{n+1}=0$.

\medskip
In both the semisimple and nilpotent cases,
the secular coefficients $P_{j,m}(\varepsilon,t,A)$ are formal power series
in $\varepsilon$. For each perturbation order and each harmonic
$\e^{\I mt}$, they give rise to a polynomial in $t$. Thus,
the appearance of polynomial $t$-dependence is not specific to the
nilpotent case.
The distinction lies in the linear part.
In the semisimple setting, the unperturbed solution has the form
$y^{(0)}(t)=\mathrm{e}^{\mathrm{i} M t}A$ in \eqref{shokou},
so the factors $\mathrm{e}^{\mathrm{i} m_j t}$ act as {\em carriers} while the
amplitudes $A_j$ (replaced under renormalization by $\mathscr{A}_j(t)$)
serve as center variables.
In the nilpotent setting, the unperturbed solution
(see \eqref{yg0}) has the polynomial form
$y^{(0)}(t)=\sum_{k} A_{k+1} t^k/k!$,
so the polynomial factors $t^k$ act as carriers while the coefficients
$A_k$ (replaced under renormalization by $\mathscr{A}_k(t)$)
constitute the center variables.
The functional RG construction extracts the dynamics on these center variables in both cases.

Let us further remark on the special case of \eqref{yeq4} in which 
$V_1(\varepsilon, \e^{\pm \I t}, y) = \cdots 
= V_{n-1}(\varepsilon, \e^{\pm \I t}, y)=0$.
Then, since $dy_j/dt = y_{j+1}$ for $1 \le j < n$, the system reduces to a
single equation for $y_1$:
\begin{equation}\label{vnr}
\frac{d^ny_1}{dt^n}
= \varepsilon V_n\Bigl(
 \varepsilon, \e^{\pm \I t},
 y_1, \frac{dy_1}{dt},\ldots, \frac{d^{n-1}y_1}{dt^{n-1}}
\Bigr).
\end{equation}
The renormalized amplitudes \eqref{ra4} satisfy the corresponding relations
$d\mathscr{A}_j/dt = \mathscr{A}_{j+1}$ for $1 \le j < n$.
Consequently, the RG equation \eqref{rg32} reduces to the single equation
\begin{align}\label{rgeqn}
\frac{d^n\mathscr{A}_1}{dt^n}
= \left.
   \frac{\partial^n}{\partial s^n}
   P_{0}\!\left(
     \varepsilon, s,
     \mathscr{A}_1,
     \frac{d\mathscr{A}_1}{dt}, \ldots, \frac{d^{n-1}\mathscr{A}_1}{dt^{n-1}}
   \right)
   \right|_{s=0}
\end{align}
for $\mathscr{A}_1=\mathscr{A}_1(\varepsilon,t,A)$.
This agrees with the RG equation \eqref{rgeq2} in the general framework of Section \ref{sec:2}
with $r=d=1$, $m_1=0$ and $n_1=n=N$, where \eqref{deq0} coincides with
\eqref{vnr}.

\begin{example}[Bogdanov--Takens type system with periodic forcing]
Consider
\begin{equation}\label{eq:BT}
\frac{d}{dt}
\begin{pmatrix}
y_1\\
y_2
\end{pmatrix}
=
\begin{pmatrix}
y_2\\
0
\end{pmatrix}
+\varepsilon
\begin{pmatrix}
2\alpha y_1 \cos t\\
\beta y_2(\mu + y_1^2 + 2\cos t)
\end{pmatrix},
\end{equation}
where $\alpha,\beta,\mu \in \mathbb{R}$.
This is a special case of \eqref{yeq42} with $n=2$, in which a nonlinear and time-periodic perturbation
is applied to the unperturbed $(\varepsilon=0)$ system whose equilibrium at the origin
is of Bogdanov--Takens type with a nilpotent linear part
(see, e.g., \cite[Sec.~8.7]{Kuz}).

From \eqref{ra4}, we have $\mathcal{A}_j = P_{j,0}$, where  
\begin{align*}
P_{1,0} &= 
A_1 + A_2 t
+ \frac{\varepsilon \beta A_2 t^2}{12}\left(
6A_1^2 + 6\mu
+ 4 A_1 A_2 t
+ A_2^2 t^2
\right)+\mathcal{O}(\varepsilon^2),
\\
P_{2,0} &= 
A_2
+ \frac{\varepsilon \beta A_2 t}{3}\left(
3A_1^2 + 3\mu
+ 3A_1 A_2 t
+  A_2^2 t^2
\right)+\mathcal{O}(\varepsilon^2).
\end{align*}
Applying the nilpotent RG scheme, we find that the RG equation \eqref{rg32} 
for $(\mathscr{A}_1,\mathscr{A}_2)$ takes the form
\begin{align}
\frac{d\mathscr{A}_1}{dt}
&=
\mathscr{A}_2\left(
1
+ 2 \alpha(\alpha - \beta)\varepsilon^2
- 8 \alpha \mathscr{A}_1 \mathscr{A}_2 (3\alpha - \beta)\beta \varepsilon^3
\right)
+ O(\varepsilon^4),
\\
\frac{d\mathscr{A}_2}{dt}
&=
\beta \mathscr{A}_2\left(
\varepsilon (\mathscr{A}_1^2 + \mu)
+ 2 \left(
\alpha^2 \mathscr{A}_1^2
+ \alpha^2 \mathscr{A}_2^2
+ 4 \alpha \mathscr{A}_2^2 \beta
- \mathscr{A}_2^2 \beta^2
\right)\varepsilon^3
\right)+ O(\varepsilon^4).
\end{align}
Thus we obtain a genuinely two-dimensional RG system.
\end{example}

\section{Scalar $N$th-Order ODE}\label{sec:2}

We consider a single $N$th-order ordinary differential equation for $y = y(t)$ of the form
\begin{align}\label{deq0}
  \Bigl(\frac{d}{dt} - \I m_1\Bigr)^{\!n_1}\cdots \Bigl(\frac{d}{dt} - \I m_d\Bigr)^{\!n_d} y
  = \varepsilon\, V\Bigl(\varepsilon, \e^{\pm \I t}, y, \frac{dy}{dt}, \ldots, \frac{d^{N-1}y}{dt^{N-1}}\Bigr),
\end{align}
where $n_1, \ldots, n_d$ are positive integers such that $N = n_1 + \cdots + n_d$ for some $d \ge 1$, 
and $m_1, \ldots, m_d$ are \emph{distinct} integers.
The equation (\ref{deq0}) contains a parameter $\varepsilon$ 
with respect to which a perturbation series is to be constructed.
The function $V$ is assumed to be a polynomial in the variables  
$\varepsilon$, $y$, $\frac{dy}{dt}$, $\ldots$, $\frac{d^{N-1}y}{dt^{N-1}}$, and a Laurent polynomial in $\e^{\I t}$.  

\subsection{Naive perturbation}\label{sec:nap}
We set
$y = y(\varepsilon, t) = \sum_{k \ge 0} \varepsilon^k y_k(t)$,
and seek a solution of the form 
\begin{align}\label{ysol1}
y(\varepsilon,t) = \sum_{m \in \Z}  \sum_{k\in \Z_{\ge 0}} \varepsilon^k f_{m,k}(t) \e^{\I m t},
\qquad f_{m,k}(t) \;\; \text{a polynomial in $t$}.
\end{align}
This is a formal power series in $\varepsilon$, and also a formal Laurent series in $\e^{\I t}$, where 
the order-$\varepsilon^k$ term corresponds to 
$y_k(t) = \sum_{m \in \Z} f_{m,k}(t) \e^{\I m t}$.
The equation for each order of $\varepsilon$ reads
\begin{subequations}
\begin{align}
\Bigl(\frac{d}{dt} - \I m_1\Bigr)^{\!n_1}\cdots \Bigl(\frac{d}{dt} - \I m_d\Bigr)^{\!n_d}y_0 &= 0,
\label{yeq00} \\
\Bigl(\frac{d}{dt} - \I m_1\Bigr)^{\!n_1}\cdots \Bigl(\frac{d}{dt} - \I m_d\Bigr)^{\!n_d} y_1 
&= V\Bigl(0, \e^{\pm \I t}, y_0, \frac{dy_0}{dt}, \ldots, \frac{d^{N-1}y_0}{dt^{N-1}} \Bigr), \nonumber\\
& \vdots \nonumber \\
\Bigl(\frac{d}{dt} - \I m_1\Bigr)^{\!n_1}\cdots \Bigl(\frac{d}{dt} - \I m_d\Bigr)^{\!n_d}y_k 
&= \left[ V\Bigl(\varepsilon, \e^{\pm \I t}, 
\sum_{j=0}^{k-1} \varepsilon^j y_j,
\sum_{j=0}^{k-1} \varepsilon^j \frac{dy_j}{dt}, \ldots,
\sum_{j=0}^{k-1} \varepsilon^j \frac{d^{N-1}y_j}{dt^{N-1}}
\Bigr) \right]_{\varepsilon^{k-1}}.
\label{yeqn2}
\end{align}
\end{subequations}
where the notation $[X ]_{\alpha}$ is defined after \eqref{yeqk}.
A general solution to \eqref{yeq00} is given by
\begin{align}\label{yA1}
y_0(t) &= \sum_{r=1}^d \sum_{j=1}^{\!n_r} \frac{A_{r,j}t^{j-1}}{(j-1)!} \e^{\I m_r t}
\end{align}
in terms of arbitrary parameters $A_{r,j}$ with $1 \le r \le d$, $1 \le j \le n_r$.
For later convenience, we group them and use the following notation:
\begin{align}
 &\mathcal{I} = \{(r,j)\mid 1 \le r \le d, 1 \le j \le n_r\},
\label{Idef}\\
&A =(A_{r,j})_{(r,j) \in \mathcal{I}} = (A_1,\ldots, A_d), \quad A_r = (A_{r,1},\ldots, A_{r,n_r})
\quad (1 \le r \le d),
\label{Adef}\\
&y_0(t) = y_0(t,A) = \sum_{r=1}^d h_r(t,A_r) \e^{\I m_r t},
\quad 
h_r(t, A_r) = \sum_{j=1}^{n_r} \frac{A_{r,j}t^{j-1}}{(j-1)!}\quad (1 \le r \le d).
\label{hrdef}
\end{align}

Starting from $y_0(t,A)$ in (\ref{hrdef}), one can construct $y_k(t)$ successively for $k=1,2,\ldots$ 
by adding a special solution to the full inhomogeneous equation (\ref{yeqn2}) 
and a general solution to the homogeneous equation (i.e., the same equation without $V$).
It is {\em uniquely} determined by imposing the condition:
\begin{align}\label{cond1}
\left. \frac{\partial^{j-1}}{\partial t^{j-1}}[y_k(t)]_{\e^{\I m_r t}} \right|_{t=0} = 0 \qquad ((r,j) \in \mathcal{I}).
\end{align}
In the notation 
$y_k(t) = \sum_{m \in \Z}f_{m,k}(t)\e^{\I m t}$ from (\ref{ysol1}), we have 
$[y_k(t)]_{\e^{\I m_r t}}= f_{m_r,k}(t)$.
The condition (\ref{cond1}) then implies that the $|\mathcal{I}|$ 
parameters in the general solution for each $y_k(t)\, (k \ge 1)$ should be chosen so that 
the  polynomial $f_{m_r,k}(t)$ is divisible by $t^{n_r}$ for all $1 \le r \le d$.

For $A=(A_{r,j})_{(r,j)\in\mathcal{I}}$,  
let $Y(\varepsilon,t,A)$ denote the unique formal solution constructed above, 
and define $P_m(\varepsilon,t,A)$ as the coefficient in its harmonic expansion:
\begin{align}\label{yp0}
Y(\varepsilon,t,A)
 = \sum_{m\in\Z} P_m(\varepsilon,t,A)\,\e^{\I m t}.
\end{align}
We refer to each $P_m(\varepsilon, t, A)$ as a \emph{secular coefficient}.  
It is a formal power series in $\varepsilon$ and $t$, and corresponds to 
$\sum_{k \ge 0}\varepsilon^k f_{m,k}(t)$ in the  notation of (\ref{ysol1}).
By definition, the secular coefficients have the following behavior:
\begin{equation}\label{co1}
P_m(\varepsilon, t, A) 
= \sum_{r=1}^d\bigl(h_r(t,A_r)\delta_{m, m_r} + \mathcal{O}(\varepsilon t^{n_r\delta_{m, m_r}})\bigr).
\end{equation}
The case $m = m_r$ in $P_m(\varepsilon, t, A)$ is referred to as the 
\emph{resonant secular coefficient}, and plays an important role.  
For instance, \eqref{co1} implies
\begin{align}
&P_m(0, t, A) = h_1(t,A_1)\delta_{m, m_1} + \cdots + h_d(t,A_d)\delta_{m, m_d},
\label{pe1}\\
&\left. \frac{\partial^{j-1}}{\partial t^{j-1}}P_{m_r}(\varepsilon, t, A)\right|_{t=0} = A_{r,j}\qquad ((r,j) \in \mathcal{I}).
\label{pt1}
\end{align}

From the successive construction of $y_0, y_1, \ldots$ and the assumption that $V$ is a Laurent polynomial in $\e^{\I t}$,  
it follows that the secular coefficients satisfy the following property:
\begin{align}\label{ped1}
P_m(\varepsilon, t, A) = \mathcal{O}(\varepsilon^{d_m}),  
\quad d_m \to \infty \quad \text{as} \quad |m| \to \infty.
\end{align}

\subsection{Functional equation for secular coefficients}\label{sec:fun}

The following lemma is a direct consequence of the preceding arguments.
\begin{lemma}\label{lem:con1}
Let s be an arbitrary parameter. 
The formal power series $y(\varepsilon,t)$ of the form (\ref{ysol1}) that satisfies the conditions (i), (ii), and (iii) below 
is unique, and coincides with $Y(\varepsilon,t,A)$ defined in (\ref{yp0}).
\begin{align*}
(\I)&\; y(\varepsilon,t)\ \text{satisfies (\ref{deq0})}, \\
(\I\I)&\; y(0,t)=y_0(t,A) \;\text{in \,(\ref{hrdef})} ,\\
(\I\I\I)&\; \frac{\partial^{j-1}([y(\varepsilon,t)]_{\e^{\I m_r t}})}{\partial t^{j-1}}\Bigr|_{t=s}
=\frac{\partial^{j-1}P_{m_r}(\varepsilon,s,A)}{\partial s^{j-1}}
\qquad ((r,j) \in \mathcal{I}).
\end{align*}
\end{lemma}

The following lemma is an analogue of Lemma \ref{le:ren2} and Lemma \ref{le:ren5}.

\begin{lemma}\label{lem:arb}
For arbitrary parameters s and $B=(B_{r,j})_{(r,j) \in \mathcal{I}}$, the formal power series
\begin{align}
\sum_{m\in\Z}P_{m}(\varepsilon,t-s,B)\e^{\I mt}
\end{align}
is also a solution to the differential equation (\ref{deq0}).
\end{lemma}
\begin{proof}
The proof is similar to the one for Lemma \ref{le:ren2}.
The substitution of $A$ with arbitrary $B$ poses no issue, 
since the dependence on initial parameters is entirely formal. 
The essential point is that shifting $t$ to $t-s$ in $P_m$  
without changing $\e^{\I mt}$ into $\e^{\I m(t-s)}$ keeps it a solution of (\ref{deq0}). 
To show this, regard (\ref{yp0}) as a formal Laurent series in $\e^{\I t}$.
By substituting it into (\ref{deq0}) 
and taking the coefficient of $\e^{\I mt}$, we get an infinite system of equations among the secular coefficients:
\begin{equation}\label{peq1}
\begin{split}
&\Bigl(\frac{\partial}{\partial t}+\I m-\I m_1\Bigr)^{\!n_1}
\cdots 
\Bigl(\frac{\partial}{\partial t}+\I m-\I m_d\Bigr)^{\!n_d} P_m(\varepsilon,t,A)\\
&=\left[ \varepsilon V\left(\varepsilon,\e^{\pm\I t},\sum_{l\in\Z}P_l(\varepsilon,t,A)\e^{\I lt},
\ldots,\sum_{l\in\Z}\frac{t^{N-1}}{\partial t^{N-1}}P_l(\varepsilon,t,A)\e^{\I lt}\right) \right]_{\e^{\I mt}}
\end{split}
\end{equation}
In general, the RHS involves infinite sums. However, thanks to (\ref{ped1}), they are actually finite at each order 
of $\varepsilon$, and thus makes sense as a formal power series in $\varepsilon$.
The key observation here is that the system (\ref{peq1}) is fully autonomous. 
The variable $t$ appears only through $\{P_l(\varepsilon, t, A)\mid l \in \Z\}$,
regardless of whether the original equation (\ref{deq0}) is autonomous or not.
This is a direct consequence of the assumption that $V$ 
is a polynomial in the variables  
$\varepsilon$, $y$, $\frac{dy}{dt}$, $\ldots$, $\frac{d^{N-1}y}{dt^{N-1}}$, 
and a Laurent polynomial in $\e^{\I t}$.  
Therefore, replacing $t$ with $t-s$ in $P_m(\varepsilon,t,A)$ yields another formal power series solution to (\ref{deq0}),
as claimed.
\end{proof}

We refer to $A=(A_{r,j})_{(r,j) \in \mathcal{I}}$ as the {\em bare amplitudes},  
and define the {\em renormalized amplitudes} as follows:
\begin{equation}\label{Ardef}
\begin{split}
&\mathscr{A}(\varepsilon,t,A) 
= (\mathscr{A}_{r,j}(\varepsilon,t,A))_{(r,j) \in \mathcal{I}}
= (\mathscr{A}_1(\varepsilon,t,A), \ldots, \mathscr{A}_d(\varepsilon,t,A)), 
\\
&\mathscr{A}_r(\varepsilon,t,A) = 
(\mathscr{A}_{r,1}(\varepsilon,t,A), \ldots, \mathscr{A}_{r,n_r}(\varepsilon,t,A)), 
\quad  \mathscr{A}_{r,j}(\varepsilon,t,A):=
\frac{\partial^{j-1}P_{m_r}(\varepsilon,t,A)}{\partial t^{j-1}}.
\end{split}
\end{equation}
By definition, they satisfy the simple relation
\begin{align}\label{pe11}
\frac{\partial}{\partial t}\mathscr{A}_{r,j}(\varepsilon,t,A) = 
\mathscr{A}_{r,j+1}(\varepsilon,t,A)\quad (1 \le r \le d, 1 \le j \le n_r-1).
\end{align}
Therefore, the renormalized amplitudes are
determined from 
$\mathscr{A}_{r,1}(\varepsilon,t,A)\, (1 \le r \le d)$ as
\begin{align}\label{AA}
\mathscr{A}(\varepsilon,t,A) =\left(
\frac{\partial^{j-1}\mathscr{A}_{r,1}(\varepsilon,A)}{\partial t^{j-1}}\right)_{(r,j) \in \mathcal{I}}.
\end{align}
We refer to
$\mathscr{A}_{1,1}(\varepsilon,t,A), \ldots, \mathscr{A}_{d,1}(\varepsilon,t,A)$ 
as the  {\em basic renormalized amplitudes}.
Our forthcoming RG equation (\ref{rgeq2}) will be described 
in terms of them.

The main result of this section is the following theorem and its consequences.
\begin{theorem}\label{thm:main1}
For any $s,t$,and $A=(A_1,\ldots,A_d)$, the following equality holds:
\begin{align}\label{func}
\sum_{m\in\Z}P_m(\varepsilon,t,A)\e^{\I mt}
=\sum_{m\in\Z}P_m(\varepsilon,t-s,\mathscr{A}(\varepsilon,s,A))\e^{\I mt}. 
\end{align}
\end{theorem}

\begin{proof}
It suffices to verify that the RHS of (\ref{func})  satisfies the conditions (i), (ii), and (iii) in Lemma \ref{lem:con1}.
Condition (i) follows from Lemma \ref{lem:arb}.
Condition (ii) is shown as 
\begin{equation}
\begin{split}
&\sum_{m \in \Z} P_m\left(0,t-s, \mathscr{A}(0,s,A)\right)\e^{\I mt}
\overset{(\ref{pe1})}{=}\sum_{r=1}^dh_r(t-s,\mathscr{A}_r(0,s,A))\e^{\I m_r t}
\overset{(\ref{hrdef})}{=}\sum_{r=1}^d\e^{\I m_r t}\sum_{j=1}^{n_r} \frac{(t-s)^{j-1}}{(j-1)!}\mathscr{A}_{r,j}(0,s,A)
\nonumber\\
&\overset{(\ref{Ardef})}{=}\sum_{r=1}^d\e^{\I m_r t}
\sum_{j=1}^{n_r} \frac{(t-s)^{j-1}}{(j-1)!}\frac{\partial^{j-1}P_{m_r}(0,s,A)}{\partial s^{j-1}}
\overset{(\ref{pe1})}{=}\sum_{r=1}^d\e^{\I m_r t}
\sum_{j=1}^{n_r} \frac{(t-s)^{j-1}}{(j-1)!} \frac{\partial^{j-1}h_r(s,A_r)}{\partial s^{j-1}} 
=\sum_{r=1}^dh_r(t,A_r)\e^{\I m_r t} ,
\end{split}
\end{equation}
where the last step uses the fact that $h_r(s,A_r)$ in (\ref{hrdef}) is a polynomial in $s$ with degree at most $n_r-1$.
Condition (iii) is shown as
\begin{equation}
\begin{split}
&\frac{\partial^{j-1}}{\partial t^{j-1}}\left.\left(\Biggl[\sum_{m \in \Z}P_m(\varepsilon,t-s,\mathscr{A}(\varepsilon,s,A))
\e^{\I m t}\Biggr]_{\e^{\I m_r t}}\right)\right|_{t=s}
=
\left. \frac{\partial^{j-1}P_{m_r}(\varepsilon,t-s,\mathscr{A}(\varepsilon,s,A))}{\partial t^{j-1}}\right|_{t=s}
\nonumber\\
&\overset{(\ref{pt1})}{=}
\mathscr{A}_{r,j}(\varepsilon,s,A)
\overset{(\ref{Ardef})}{=}\frac{\partial^{j-1}P_{m_r}(\varepsilon,s,A)}{\partial s^{j-1}}.
\end{split}
\end{equation}
\end{proof}

\begin{corollary}\label{cor:fun}
The secular coefficients satisfy the following functional relation:
\begin{align}\label{PP1}
P_m(\varepsilon,t,A)
=P_m(\varepsilon,t-s,\mathscr{A}(\varepsilon,s,A))\quad (m \in \Z).
\end{align}
\end{corollary}

In particular, setting $m=m_r$ and taking derivative $(\partial/\partial t)^{j-1}$ yields the 
functional relation for renormalized amplitudes defined in (\ref{Ardef}):
\begin{align}\label{fun2}
\mathscr{A}_{r,j}(\varepsilon,t,A)=\mathscr{A}_{r,j}(\varepsilon,t-s,\mathscr{A}(\varepsilon,s,A))
\quad ((r,j) \in \mathcal{I}).
\end{align}
By replacing $(t,s)$ with $(t+s,t)$ in (\ref{fun2}) and using (\ref{Ardef}), we obtain
\begin{align}
\mathscr{A}_{r,j}(\varepsilon,t+s,A)
=\frac{\partial^{j-1}}{\partial s^{j-1}}P_{m_r}(\varepsilon,s,\mathscr{A}(\varepsilon,t,A))
\quad ((r,j) \in \mathcal{I}).
\end{align}
Differentiating this identity with respect to $s$ at $s=0$ gives
\begin{align}\label{rg1}
\frac{d}{dt}\mathscr{A}_{r,j}(\varepsilon,t,A)
=\frac{\partial^{j-1}}{\partial s^{j-1}}P_{m_r}(\varepsilon,s,\mathscr{A}(\varepsilon,t,A))\Big|_{s=0}
\quad ((r,j) \in \mathcal{I}).
\end{align}
By (\ref{pt1}), the RHS equals $\mathscr{A}_{r,j+1}(\varepsilon,t,A)$ for $1 \le j \le n_r-1$. 
Therefore, (\ref{rg1}) is trivially satisfied in this range of $j$ due to  (\ref{pe11}).
The RG equation governing the dynamics of the renormalized amplitudes is obtained from 
the unique nontrivial case $j=n_r$:
\begin{align}\label{rgeq2}
\frac{d^{n_r}}{dt^{n_r}}\mathscr{A}_{r,1}(\varepsilon,t,A) 
= \left.\frac{\partial^{n_r}}{\partial s^{n_r}}
P_{m_r}(\varepsilon, s,\mathscr{A}(\varepsilon,t,A))\right|_{s=0}
\quad (1 \le r \le d).
\end{align}
In view of (\ref{AA}), this constitutes a system of $n_r$th-order autonomous ordinary differential equations 
for the basic renormalized amplitudes $\mathscr{A}_{r,1}(\varepsilon,t,A) \, (1 \le r \le d)$.

Suppose that the secular coefficients are obtained by naive perturbation as
\begin{align}
P_{m}(\varepsilon,t,A)
&= h_1(t,A_1)\delta_{m,m_1}
 + \cdots
 + h_d(t,A_d)\delta_{m,m_d} 
 + \sum_{k,l\ge 1,\, i_{r,j}\ge 0}
   D_{m;k,l,\{i_{r,j}\}}
   \varepsilon^k t^{\,l-1}
   \prod_{(r,j)\in\mathcal{I}}
   A_{r,j}^{\,i_{r,j}},
\end{align}
where the coefficients $D_{m_h;k,l,\{i_{r,j}\}}$ vanish for $(h,l)\in\mathcal{I}$, 
reflecting \eqref{co1}.
Then, from \eqref{pe11}, the RG equation \eqref{rgeq2} becomes
\begin{align}\label{rg7}
\frac{d^{n_r}\mathscr{A}_{r,1}}{dt^{n_r}}
&=
\sum_{k\ge 1,\, i_{r,j}\ge 0}
 \varepsilon^k\, n_r!\,
 D_{m_r;k,n_r+1,\{i_{r,j}\}}
 \prod_{(r,j)\in\mathcal{I}}
 \left(\frac{\partial^{j-1}\mathscr{A}_{r,1}}{\partial t^{j-1}}\right)^{\!i_{r,j}}
\qquad (1\le r\le d),
\end{align}
where $\mathscr{A}_{r,1}=\mathscr{A}_{r,1}(\varepsilon,t,A)$.
The RHS of (\ref{rg7})  is of order $\mathcal{O}(\varepsilon)$ which is also readily seen from (\ref{co1}).
In this sense, the dynamics (\ref{rg7}) is slow compared with the original 
time evolution in (\ref{yA1}).

By applying (\ref{PP1}) with $s=t$ to the naive expansion (\ref{yp0}), we obtain the 
renormalized expansion:
\begin{align}\label{rs2}
Y(\varepsilon,t,A) 
= \sum_{m \in \Z} 
P_{m}(\varepsilon,0,\mathscr{A}(\varepsilon,t,A))\e^{\I mt}.
\end{align}

One can invert the formula (\ref{Ardef})  for the renormalized amplitudes in terms of the bare ones as 
\begin{align}\label{inv1}
A_{r,j} \overset{(\ref{pt1})}{=}
\frac{\partial^{j-1}P_{m_r}(\varepsilon,t,A)}{\partial t^{j-1}}\Big|_{t=0} 
\overset{(\ref{PP1})}{=}
\frac{\partial^{j-1}P_{m_r}(\varepsilon,t-s,\mathscr{A}(\varepsilon,s,A))}{\partial t^{j-1}}\Bigr|_{t=0}.
\end{align}
Here, one may interchange $s$ and $t$ in the final expression, since $A_{r,j}$ is actually  independent of 
either variable.

\begin{example}\label{ex:1}
Although elementary, this example provides a convenient illustration
of the main ideas and a useful check of the results obtained in the main text.
Consider the single component case 
\begin{align}\label{ex1:0}
\frac{dy}{dt}=\varepsilon V(\varepsilon,y),
\end{align}
where $V$ is autonomous, namely, a polynomial in $y$ independent of $t$.
This corresponds to the case $N=d=n_1=1, m_1=0$ in (\ref{deq0}).
The expansion (\ref{yp0}) of the solution reduces to 
a single term $y=Y(\varepsilon,t,A) = P_0(\varepsilon,t,A)$.
It is formally determined by the relation
\begin{align}\label{ex1:2}
\int^{P_0(\varepsilon,t,A)}_A\frac{dz}{V(\varepsilon,z)} = \varepsilon t,
\end{align}
satisfying $P_0(0,t,A)= P_0(\varepsilon,0,A)=A$.
Replacing $t$ with $s$, or $(t,A)$ with $(t-s, P_0(\varepsilon, s, A))$ in (\ref{ex1:2}),  
we obtain two relations:
\begin{align}\label{ex1:3}
\int^{P_0(\varepsilon,s,A)}_A\frac{dz}{V(\varepsilon,z)} = \varepsilon s, \qquad
\int^{P_0(\varepsilon,t-s,P_0(\varepsilon, s, A))}_{P_0(\varepsilon, s, A)}\frac{dz}{V(\varepsilon,z)} = \varepsilon (t-s).
\end{align}
Comparing the sum of the two relations with (\ref{ex1:2}) 
leads to the functional relation (\ref{PP1}) for the unique secular coefficient:
$P_0(\varepsilon,t,A) = P_0(\varepsilon,t-s,P_0(\varepsilon,s,A))$. 
For example, when  $V(\varepsilon,y)=y^2-1$, this relation is satisfied by 
$P_0(\varepsilon,t,A) = (A \cosh \varepsilon t - \sinh \varepsilon t)/ (A \sinh \varepsilon t + \cosh \varepsilon t)$.

The unique renormalized amplitude (\ref{Ardef}) is simply given by 
$\mathscr{A}(t) := \mathscr{A}_{1,1}(\varepsilon,t,A) = P_0(\varepsilon,t,A)$.
The RG equation (\ref{rgeq2}) reads 
\begin{align}\label{ex1:4}
\frac{d}{dt}\mathscr{A}(t) = \frac{\partial}{\partial s}P_0(\varepsilon,s,\mathscr{A}(t))\Bigr|_{s=0}.
\end{align}
To compute the RHS, differentiate the first relation in (\ref{ex1:3}) with respect to $s$ to obtain
\begin{align}
\frac{\partial}{\partial s}P_0(\varepsilon,s,A) = \varepsilon  V(\varepsilon,P_0(\varepsilon,s,A)).
\end{align} 
Setting $(s,A)=(0,\mathscr{A}(t))$ and using $P_0(\varepsilon,0,A)=A$, 
we find that (\ref{ex1:4}) becomes
\begin{align}
\frac{d}{dt}\mathscr{A}(t) =  \varepsilon V(\varepsilon, \mathscr{A}(t)).
\end{align}
Thus, in this simple setting, the RG equation simply coincides with the original equation (\ref{ex1:0}) itself,
in agreement with the remark given around \eqref{ypa}.
\end{example}

\begin{example}\label{ex:3rd}
Let us consider a third-order nonlinear, non-autonomous example:
\begin{align}
\frac{d^3y}{dt^3} = 2\varepsilon y \frac{d^2y}{dt^2}\cos t.
\end{align}
This corresponds to \eqref{deq0} with $d=1$, $m_1=0$, and $N=n_1=3$.
We simply denote the set $\mathcal{I} = \{(1,1), (1,2), (1,3)\}$ in \eqref{Idef} by $\{1,2,3\}$.
Correspondingly, the bare amplitude \eqref{Adef} is written as $A=(A_1,A_2,A_3)$, and 
the function \eqref{hrdef} as $h(t,A) = A_1 + A_2 t + A_3 t^2/2$.
The renormalized amplitude \eqref{Ardef} is given by
$\mathscr{A}=(\mathscr{A}_1,\mathscr{A}_2,\mathscr{A}_3)
=\Bigl(\mathscr{A}_1, \frac{d\mathscr{A}_1}{dt},
\frac{d^2\mathscr{A}_1}{dt^2}\Bigr)$ with 
$\mathscr{A}_1=\mathscr{A}_1(\varepsilon,t,A)=P_0(\varepsilon,t,A)$.

A straightforward calculation leads to 
\begin{equation}\label{p0}
\begin{split}
P_0(\varepsilon,t, A)
&=h(t,A)
+ \frac{\varepsilon^2A_3  t^3}{120}
\Bigl(
  40 A_2 (A_1 - 3 A_3)
  + 10 (A_2^2 + A_1 A_3 - 3 A_3^2) t
  + 6 A_2 A_3 t^2
  + A_3^2 t^3
\Bigr) + \mathcal{O}(\varepsilon^4)
\\
 &= h(t,A)+ Q(\varepsilon, A) t^3 + \mathcal{O}(t^4)
\\
Q(\varepsilon,A) 
&=\frac{\varepsilon^2 A_2A_3}{3} (A_1 - 3 A_3) 
\\
&+ \frac{\varepsilon^4A_2 A_3}{192} 
\Bigl(
  -32 A_1^3
  - 32 A_1 A_2^2
  - 120 A_1^2 A_3
  + 924 A_2^2 A_3
  + 5576 A_1 A_3^2
  - 48165 A_3^3
\Bigr)+ \mathcal{O}(\varepsilon^6).
\end{split}
\end{equation}
The RG equation \eqref{rgeq2} for the basic renormalized amplitude $\mathscr{A}_1$  reads
\begin{equation}\label{a3eq}
\begin{split}
\frac{d^3\mathscr{A}_1}{dt^3} 
&= 6 Q\Bigl(\varepsilon, \mathscr{A}_1, \frac{d\mathscr{A}_1}{dt},
\frac{d^2\mathscr{A}_1}{dt^2}\Bigr)+ \mathcal{O}(\varepsilon^6).
\end{split}
\end{equation}
The renormalized expansion \eqref{rs2}, free of secular terms, reads
\begin{equation*}
\begin{split}
Y(\varepsilon,t,A) 
&= 
\mathscr{A}_1
- 2 \varepsilon \mathscr{A}_3 
\bigl(
  3 \mathscr{A}_2 \cos t
  + (\mathscr{A}_1 - 6 \mathscr{A}_3) \sin t
\bigr)
\\
&+ \frac{\varepsilon^2 \mathscr{A}_3}{32} 
\Bigl(
  (8 \mathscr{A}_1^2 - 36 \mathscr{A}_2^2 - 52 \mathscr{A}_1 \mathscr{A}_3 + 183 \mathscr{A}_3^2)\cos(2t)
  + 16 \mathscr{A}_2 (-2 \mathscr{A}_1 + 9 \mathscr{A}_3)\sin(2t)
\Bigr) + \mathcal{O}(\varepsilon^3),
\end{split}
\end{equation*}
where $\mathscr{A}_2 = d\mathscr{A}_1/d t$ and 
$\mathscr{A}_3 = d^2\mathscr{A}_1/d t^2$, and 
$\mathscr{A}_1$ satisfies \eqref{a3eq}.
\end{example}

\subsection{A difference equation}\label{ss:de}

As a miscellaneous example, we study a difference equation that may be
formally regarded as an $N=\infty$ case in which
$\{m_1,\ldots,m_N\}\rightarrow\Z$.
Consider \eqref{deq0} with $N=d=2K+1$, $n_1=\cdots=n_N=1$, and
$\{m_1,\ldots,m_N\}=\{0,\pm1,\ldots,\pm K\}$ for a positive integer $K$.
With a suitable $K$-dependent rescaling of $V$, the equation takes the form
\begin{equation}\label{ex21}
2\pi\frac{d}{dt}\prod_{n=1}^K
\left(1+\frac{1}{n^2}\frac{d^2}{dt^2}\right)y
=
\varepsilon V\!\left(
\varepsilon,\e^{\pm\I t},y,\frac{dy}{dt},\ldots,\frac{d^{2K}y}{dt^{2K}}
\right).
\end{equation}
Using the infinite product representation
$\sin z=z\prod_{n=1}^{\infty}\left(1-\frac{z^2}{n^2\pi^2}\right)$,
the formal limit $K\rightarrow\infty$ of the differential operator on the
LHS yields the difference operator
$2\sinh\!\left(\pi\frac{d}{dt}\right)$.
Consequently, one obtains the difference-differential equation
\begin{align}\label{ex22}
y(t+\pi)-y(t-\pi)
=
\varepsilon V\!\left(
\varepsilon,\e^{\pm\I t},y,\frac{dy}{dt},\ldots
\right).
\end{align}
In what follows, we consider a simple example in which $V$ is linear in $y$:
\begin{subequations}
\begin{align}
y(t+\pi)-y(t-\pi)
&=
2\varepsilon U(\e^{\I t})y(t),
\label{yyu}
\\
2U(z)
&=
\sum_{l}\alpha_l z^l,
\label{ua}
\end{align}
\end{subequations}
where the factor $2$ is introduced only to simplify subsequent formulas.
The sum runs over $l\in\Z$, and the coefficients $\alpha_l$ are arbitrary
except that $\alpha_l=0$ for all but finitely many $l\in\Z$.
Up to a slight modification $y(t)\rightarrow\e^{\I t/2}y(t)$, the equation
\eqref{yyu} appears in several contexts in physics and mathematics,
including the Harper equation in the Hofstadter problem \cite{Harper,Hof}
and the Baxter $TQ$ relation in quantum integrable systems (cf. \cite[eq.~(9.4.5)]{Baxter}).

In the present setting, the entire set $\{\e^{\I m t}\mid m\in\Z\}$
consists of resonant harmonics.
Accordingly, the unperturbed solution takes the form
$y_0(t)=\sum_{m\in\Z}A_m\e^{\I m t}$, where
$A=\{A_m\mid m\in\Z\}$ are arbitrary coefficients playing the role of
bare amplitudes.

The order-$k$ perturbation equation corresponding to \eqref{yeqn2} reads
\begin{equation}
y_k(t+\pi)-y_k(t-\pi) = 2U(\e^{\I t})y_{k-1}(t)\qquad (k\ge 1).
\end{equation}
Substituting 
$y_k(t) = \sum_{m \in \Z} f_{m,k}(t) \e^{\I m t}$ and comparing the 
coefficient of $\e^{\I mt}$, we obtain 
\begin{equation}\label{feq}
(-1)^m \bigl( f_{m,k}(t+ \pi) - f_{m,k}(t-\pi)\bigr)
=\sum_l \alpha_l f_{m-l,k-1}(t)  \quad (k \ge 1, m \in \Z),
\end{equation}
with the initial condition $f_{m,0}(t) = A_m$ with respect to $k$.
The condition \eqref{cond1} reads $f_{m,k}(0)=\delta_{k,0}A_m$ for all $m \in \Z$.
The polynomial $f_{m,k}(t)$ in $t$ that satisfies \eqref{feq} and these conditions is obtained in the form 
\begin{align}\label{fgb}
f_{m,k}(t) = (-1)^{mk}g_k\Bigl(\frac{t}{\pi}\Bigr)B_{m,k},
\end{align}
where $g_k(u)$ is a degree-$k$ polynomial in $u$ given by 
\begin{align}\label{gex}
g_k(u) = \frac{u\Gamma\bigl(\frac{u+k}{2}\bigr)}{2k! \Gamma\bigl(\frac{u-k}{2}+1\bigr)} 
= 1,\frac{u}{2},  \frac{u^2}{8},
\frac{u(u^2-1)}{48}, 
\frac{u^2(u^2-4)}{384},
\frac{u(u^2-1)(u^2-9)}{3840}, \ldots
\end{align}
for $k=0,1,2,3,4,5, \ldots$.
It is the solution to the recursion relation and the initial condition:
\begin{align}\label{grec}
g_k(u+1)-g_k(u-1) = g_{k-1}(u), \quad g_0(u)=1.
\end{align}
We will use the generating function and its consequences:
\begin{align}
&\sum_{k\ge 0}g_k(u)(2\zeta)^k = (\sqrt{1+\zeta^2}+\zeta)^u,
\qquad
\sum_{k\ge 0}(-1)^k \mathcal{N}_k\zeta^{2k+1} = \log(\sqrt{1+\zeta^2}+\zeta),
\label{gp1}\\
&g'_{2k}(0)=0, \quad
 g'_{2k+1}(0) = \frac{(-1)^k}{2^{2k+1}}\mathcal{N}_k,
\quad 
\mathcal{N}_k = \frac{1}{2^{2k}(2k+1)}\binom{2k}{k}
=1,\frac{1}{6}, \frac{3}{40}, \frac{5}{112}, \frac{35}{1152},\ldots
\label{gp2}
\end{align} 
for $k=0,1,2,3,4,\ldots$.

Returning to \eqref{fgb}, $B_{m,k}$ is a linear combination of the bare amplitudes $\{A_m\}$
independent of $t$.
Substituting \eqref{fgb} into \eqref{feq} and using \eqref{grec}, we have
\begin{align}\label{brec} 
B_{m,k} = \sum_l \alpha_l (-1)^{(k-1)l}B_{m-l,k-1},\qquad B_{m,0}=A_m,
\end{align}
where the latter initial condition follows from \eqref{fgb} and $f_{m,0}(t) = A_m$.
Setting 
\begin{align}\label{bca}
B_{m,k} = \sum_jC_{k,j}A_{m+j},
\end{align}
and comparing the coefficient of $A_m$ on both sides of \eqref{brec}, we obtain the recursion relation
\begin{align}
C_{k,j} = \sum_{i,l;\; i-l=j}(-1)^{(k-1)l}\alpha_l C_{k-1,i}.
\end{align}
This can be rewritten as a recursion relation for the generating function
\begin{align}\label{hrec}
h_k(z) &= \sum_j C_{k,j}z^j, \qquad 
h_k(z) = 2U\bigl((-1)^{k-1}z^{-1}\bigr) h_{k-1}(z), \qquad h_0(z) = 1.
\end{align} 
As a consequence, we obtain
\begin{align}\label{ch}
C_{k,j} &= \oint\frac{h_k(z)\,dz}{2\pi i\, z^{j+1}},
\qquad
\begin{cases}
\,h_{2k}(z) = \bigl(2U(z^{-1})U(-z^{-1})\bigr)^k,
\\
\,h_{2k+1}(z) = \bigl(2U(z^{-1})\bigr)^{k+1}\bigl(2U(-z^{-1})\bigr)^k.
\end{cases}
\end{align}

Combining the results so far, the resonant secular coefficient
$P_m(\varepsilon,t,A)=\sum_{k\ge 0}\varepsilon^k f_{m,k}(t)$ is expressed as
\begin{equation}\label{pgh}
\begin{split}
P_m(\varepsilon,t,A)
&=\sum_j A_{m+j}\sum_{k\ge 0}
\varepsilon^k(-1)^{mk}g_k\Bigl(\frac{t}{\pi}\Bigr)
\oint\frac{h_k(z)\,dz}{2\pi i\, z^{j+1}},
\\
&=A_m + \sum_jA_{m+j}\left(
\frac{\varepsilon(-1)^mt}{2\pi}\alpha_{-j}
+ \frac{\varepsilon^2 t^2}{8\pi^2}\sum_l (-1)^l\alpha_l\alpha_{-j-l}
+\mathcal{O}(\varepsilon^3)\right).
\end{split}
\end{equation}
Since the present setting formally corresponds to $n_r=1$ for all $r\in\Z$ in \eqref{deq0},
the renormalized amplitudes are simply the secular coefficients themselves.
Henceforth, we set $\mathscr{A}_m(t)=\mathscr{A}_m(\varepsilon,t,A)=P_m(\varepsilon,t,A)$.

In the rest of this subsection, we assume that $U(z)$ is an even Laurent polynomial,
namely, $U(z)=U(-z)$.
Then $h_k(z)=\bigl(2U(z^{-1})\bigr)^k$.
Using \eqref{gp1} and \eqref{pgh}, the renormalized amplitude is computed as
\begin{align}
\mathscr{A}_m(\varepsilon, t,A)
&=\sum_j A_{m+j}
\oint\frac{dz}{2\pi i\, z^{j+1}}
\Biggl(
\sum_{\substack{k \ge 0\\ \text{even}}}\varepsilon^k g_k\Bigl(\frac{t}{\pi}\Bigr)
\bigl(2U(z^{-1})\bigr)^k
+(-1)^m\sum_{\substack{k \ge 1\\ \text{odd}}}\varepsilon^k g_k\Bigl(\frac{t}{\pi}\Bigr)
\bigl(2U(z^{-1})\bigr)^k
\Biggr)
\nonumber\\
&=\sum_j A_{m+j}
\oint\frac{dz}{2\pi i\, z^{j+1}}
\left(
\cosh\Bigl(\frac{\Theta(\varepsilon,z^{-1})\, t}{\pi}\Bigr)
+(-1)^m\sinh\Bigl(\frac{\Theta(\varepsilon,z^{-1})\, t}{\pi}\Bigr)
\right)
\nonumber \\
&=\sum_j A_{m+j}
\oint\frac{dz}{2\pi i\, z^{j+1}}
\exp\Bigl(\frac{(-1)^m\Theta(\varepsilon,z^{-1})\, t}{\pi}\Bigr),
\label{csh}
\end{align}
where the function $\Theta(\varepsilon,\zeta)$ is defined by
\begin{align}\label{utt}
\Theta(\varepsilon,\zeta)
= \log\left(\sqrt{1+\varepsilon^2 U(\zeta)^2}+ \varepsilon U(\zeta)\right),
\qquad
\sinh \Theta(\varepsilon,\zeta) = \varepsilon U(\zeta).
\end{align}
An explicit power series expansion in $\varepsilon$ with vanishing initial term
$\Theta(0,\zeta)=0$ follows from the latter identity in \eqref{utt}.
Note also that $\Theta(\varepsilon,-\zeta)=\Theta(\varepsilon,\zeta)$
due to the assumption $U(z)=U(-z)$.

With the explicit formula \eqref{csh} at hand,  one can prove the functional relation (see \eqref{fun2})
\begin{align}\label{aa}
\mathscr{A}_m(\varepsilon,t, A)
= \mathscr{A}_m\!\left(\varepsilon,t-s,\{\mathscr{A}_j(\varepsilon,s, A)\}\right)
\end{align}
for $A=\{A_l\mid l \in \Z\}$.
In fact, the RHS of \eqref{aa} is calculated as
\begin{equation*}
\begin{split}
&\sum_j \mathscr{A}_{m+j}(\varepsilon,s,A)
\oint\frac{dz}{2\pi i\, z^{j+1}}
\exp\Bigl(\frac{(-1)^m\Theta(\varepsilon,z^{-1})\,(t-s)}{\pi}\Bigr)
\\
&= \sum_{j,l}A_{m+j+l}
\oint\oint\frac{dwdz}{(2\pi i)^2\, w^{l+1}z^{j+1}}
\left(
\cosh\Bigl(\frac{\Theta(\varepsilon,w^{-1})\, s}{\pi}\Bigr)
+(-1)^{m+j}\sinh\Bigl(\frac{\Theta(\varepsilon,w^{-1})\, s}{\pi}\Bigr)
\right)\\
&\qquad\qquad\qquad\qquad
\times \exp\Bigl(\frac{(-1)^m\Theta(\varepsilon,z^{-1})\,(t-s)}{\pi}\Bigr)
\\
&=\sum_n A_{m+n}\oint\oint\frac{dwdz}{(2\pi i)^2\, w^{n+1}z}
\sum_j
\left(\Bigl(\frac{w}{z}\Bigr)^j
\cosh\Bigl(\frac{\Theta(\varepsilon,w^{-1})\, s}{\pi}\Bigr)
+(-1)^{m}\Bigl(-\frac{w}{z}\Bigr)^j\sinh\Bigl(\frac{\Theta(\varepsilon,w^{-1})\, s}{\pi}\Bigr)
\right)\\
&\qquad\qquad\qquad\qquad
\times \exp\Bigl(\frac{(-1)^m\Theta(\varepsilon,z^{-1})\,(t-s)}{\pi}\Bigr).
\end{split}
\end{equation*}
The sums over $j$ yield delta functions and the subsequent integral 
$\oint dz/(2\pi i z)$ picks out the integrands at $z=\pm w$.
Using $\Theta(\varepsilon,-\zeta)=\Theta(\varepsilon,\zeta)$ further, 
the result is expressed as
\begin{equation}
\begin{split}
&\sum_n A_{m+n}
\oint\frac{dw}{2\pi i\, w^{n+1}}
\left(
\cosh\Bigl(\frac{\Theta(\varepsilon,w^{-1})\, s}{\pi}\Bigr)
+(-1)^{m}\sinh\Bigl(\frac{\Theta(\varepsilon,w^{-1})\, s}{\pi}\Bigr)
\right)
\\
&\qquad\qquad\qquad\qquad
\times \exp\Bigl(\frac{(-1)^m\Theta(\varepsilon,w^{-1})\,(t-s)}{\pi}\Bigr).
\end{split}
\end{equation}
Since the expression in parentheses equals
$\exp\!\bigl((-1)^m\Theta(\varepsilon,w^{-1})\,s/\pi\bigr)$,
this reproduces the LHS of \eqref{aa}, written in \eqref{csh}.

In view of \eqref{pgh}, the functional relation \eqref{aa} for
$\mathscr{A}_m(\varepsilon,t,A)=P_m(\varepsilon,t,A)$ yields an infinite family
of highly nontrivial identities among the coefficients of the $A_\ell$'s.
These identities can indeed be verified explicitly in low orders of
$\varepsilon$.

Let us introduce the generating series of the renormalized amplitudes \eqref{csh}:
\begin{align}
\mathscr{A}(\zeta,t)=\sum_{m}\mathscr{A}_m(\varepsilon,t,A)\zeta^m
&=
\sum_{m+j}A_{m+j}\zeta^{m+j}
\oint\frac{dz}{2\pi i z}\sum_j(\zeta z)^{-j}
\cosh\Bigl(\frac{\Theta(\varepsilon,z^{-1})\,t}{\pi}\Bigr)
\nonumber\\
&\quad+
\sum_{m+j}A_{m+j}(-\zeta)^{m+j}
\oint\frac{dz}{2\pi i z}\sum_j(-\zeta z)^{-j}
\sinh\Bigl(\frac{\Theta(\varepsilon,z^{-1})\,t}{\pi}\Bigr)
\nonumber\\
&=
\exp\Bigl(\frac{\Theta(\varepsilon,\zeta)\,t}{\pi}\Bigr)
\sum_{m:\text{even}}A_m\zeta^m
+
\exp\Bigl(-\frac{\Theta(\varepsilon,\zeta)\,t}{\pi}\Bigr)
\sum_{m:\text{odd}}A_m\zeta^m,
\label{csh2}
\end{align}
where in the last step the $j$-sums and the integrals
$\oint dz/(2\pi i z)$ pick out the integrands at $z=\pm\zeta^{-1}$.

Since $\mathscr{A}_m(\varepsilon,t,A)=P_m(\varepsilon,t,A)$ and 
$P_m(\varepsilon,0,A)=A_m$, 
the renormalized expansion
$Y(\varepsilon,t,A)=\sum_m\mathscr{A}_m(\varepsilon, t, A)\e^{\I mt}$ in \eqref{rs2}
is the formal series obtained by setting $\zeta=\e^{\I t}$ in
$\mathscr{A}(\zeta,t)$:
\begin{align}\label{fin}
Y(\varepsilon,t,A)
&=
\exp\Bigl(\frac{\Theta(\varepsilon,\e^{\I t})\,t}{\pi}\Bigr)
\sum_{m:\text{even}}A_m\e^{\I mt}
+
\exp\Bigl(-\frac{\Theta(\varepsilon,\e^{\I t})\,t}{\pi}\Bigr)
\sum_{m:\text{odd}}A_m\e^{\I mt}.
\end{align}
One can readily verify that \eqref{fin} satisfies the difference equation \eqref{yyu}:
\begin{equation}\label{yf}
\begin{split}
&Y(\varepsilon,t+\pi,A)-Y(\varepsilon,t-\pi,A)
\\
&= \mu_+(-\e^{\I t})
\exp\Bigl(\frac{\Theta(\varepsilon,-\e^{\I t})(t+\pi)}{\pi}\Bigr)
+\mu_-(-\e^{\I t})
\exp\Bigl(-\frac{\Theta(\varepsilon,-\e^{\I t})(t+\pi)}{\pi}\Bigr)
\\
&\quad-\mu_+(-\e^{\I t})
\exp\Bigl(\frac{\Theta(\varepsilon,-\e^{\I t})(t-\pi)}{\pi}\Bigr)
-\mu_-(-\e^{\I t})
\exp\Bigl(-\frac{\Theta(\varepsilon,-\e^{\I t})(t-\pi)}{\pi}\Bigr)
\\
&= \mu_+(\e^{\I t})
\exp\Bigl(\frac{\Theta(\varepsilon,\e^{\I t})(t+\pi)}{\pi}\Bigr)
-\mu_-(\e^{\I t})
\exp\Bigl(-\frac{\Theta(\varepsilon,\e^{\I t})(t+\pi)}{\pi}\Bigr)
\\
&\quad-\mu_+(\e^{\I t})
\exp\Bigl(\frac{\Theta(\varepsilon,\e^{\I t})(t-\pi)}{\pi}\Bigr)
+\mu_-(\e^{\I t})
\exp\Bigl(-\frac{\Theta(\varepsilon,\e^{\I t})(t-\pi)}{\pi}\Bigr)
\\
&= 2\sinh\Theta(\varepsilon,\e^{\I t})
\Bigl(
\mu_+(\e^{\I t})
\exp\Bigl(\frac{\Theta(\varepsilon,\e^{\I t})\,t}{\pi}\Bigr)
+\mu_-(\e^{\I t})
\exp\Bigl(-\frac{\Theta(\varepsilon,\e^{\I t})\,t}{\pi}\Bigr)
\Bigr)
\\
&= 2\varepsilon U(\e^{\I t})\,Y(\varepsilon,t,A),
\end{split}
\end{equation}
where we have used the shorthand
$\mu_\pm(\e^{\I t})
=\sum_{m\in\Z;\,(-1)^m=\pm1}A_m\e^{\I mt}$.

Under the assumption $U(-z)=U(z)$, the function $U(\e^{\I t})$ is a constant 
as long as $t$ stays within $t_0 + \Z \pi$ for a fixed $t_0$.
Accordingly, the difference equation \eqref{yyu} becomes a second-order
recursion in the discrete variable $t/\pi$, with characteristic multipliers
$\exp\!\bigl(\Theta(\varepsilon,\e^{\I t})\bigr)$ and
$\exp\!\bigl(-\Theta(\varepsilon,\e^{\I t})+\I\pi\bigr)$.
Thus \eqref{fin} can in fact be obtained by an elementary argument, by choosing
the two pseudo-constants as
$\sum_{m:\text{even}}A_m\e^{\I mt}$ and
$\sum_{m:\text{odd}}A_m\e^{\I(m-1)t}$.
The calculation above is intended to illustrate how the same conclusion arises
from the perturbative scheme in Section~\ref{sec:nap}.

The RG equation for the renormalized amplitudes $\mathscr{A}_m(t) = \mathscr{A}_m(\varepsilon,t,A)$ can be
derived, for instance, by differentiating \eqref{pgh} with respect to $t$ at
$t=0$ and proceeding similarly.
The outcome, of course, agrees with the direct differentiation of
\eqref{csh2}:
\begin{align}\label{dat}
\frac{\partial \mathscr{A}(\zeta,t)}{\partial t}
&=
\frac{\Theta(\varepsilon,\zeta)}{\pi}\,\mathscr{A}(-\zeta,t)
=
\frac{1}{\pi}
\sum_{k\ge 0}(-1)^k \varepsilon^{2k+1}\mathcal{N}_k\,
U(\zeta)^{2k+1}\mathscr{A}(-\zeta,t).
\end{align}

The expression \eqref{fin} resums the effects of $\varepsilon$ to all orders.
It suggests that the system is stable only when
$\Theta(\varepsilon,\e^{\I t})\in \I\R$ for a generic choice of the bare
amplitudes $\{A_m\}$.

\section{Conclusion}\label{sec:c}

We have developed an RG-based perturbation scheme
based on an exact functional relation among secular coefficients.
This relation provides a unified mechanism underlying
the emergence of the RG equation,
the elimination of secular terms,
and the invertibility between bare and renormalized amplitudes.
The framework applies to a class of first-order systems with semisimple
or nilpotent linear parts, as well as to scalar higher-order equations.

\appendix
\section{Numerical plots for Example \ref{ex:CD}}\label{app:n}

We consider the $y_1=y^*_2$ case of Example \ref{ex:CD}
with the parameters and the initial condition given by 
$\varepsilon= 0.25$, $(R(0), \theta(0)) = (1.3, 2.1)$.
From \eqref{ex1:yy},  it corresponds to setting $y_1(0) = y^*_2(0) = -0.462366+1.55692 \I$.

The direct numerical integration of the original ODE \eqref{ex1:eq} yields Figure \ref{fig:1}.
\begin{figure}[H]
\centering
\includegraphics[scale=0.7]{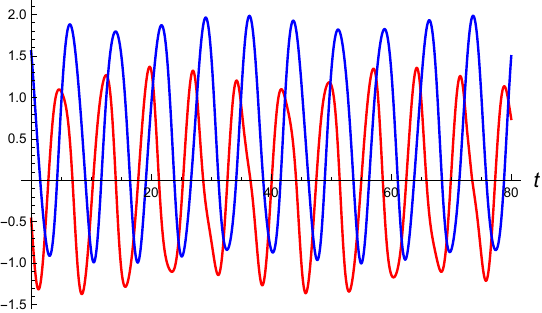}
\caption{The real part (red) and the imaginary part (blue)  of $y_1(t)=y^*_2(t)$.}
\label{fig:1}
\end{figure}

On the other hand, numerical integration of the RG equation \eqref{rteq}, 
retaining terms up to order $\mathcal{O}(\varepsilon^4)$,
yields the functions $R(t)$ and $\theta(t)$ shown in Figure \ref{fig:RT}.
\begin{figure}[H]
\centering
\begin{subfigure}{0.48\textwidth}
\centering
\includegraphics[scale=0.7]{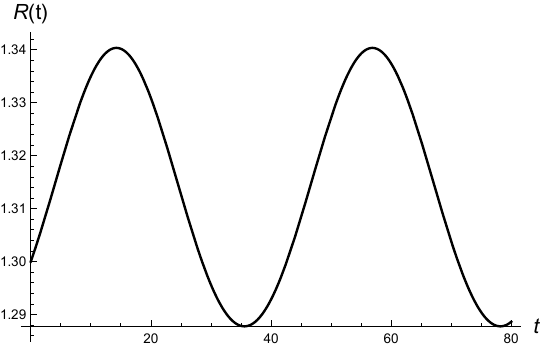}
\caption{$R(t)$}
\end{subfigure}
\hspace{0.01\textwidth}
\begin{subfigure}{0.48\textwidth}
\centering
\includegraphics[scale=0.7]{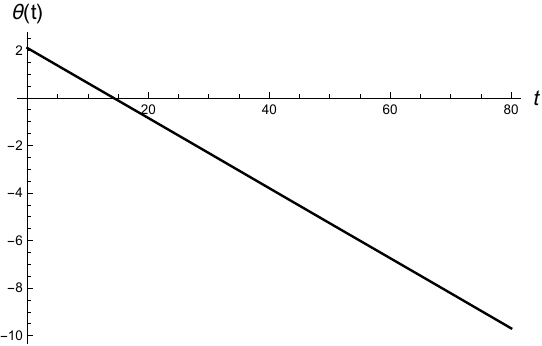}
\caption{$\theta(t)$}
\end{subfigure}

\caption{Numerical solutions of the RG equation \eqref{rteq}.}
\label{fig:RT}
\end{figure}

Finally, we compare the real part in Figure \ref{fig:1}
with that obtained from the renormalized expansion \eqref{ex1:yy}.
\begin{figure}[H]
\centering
\includegraphics[scale=0.7]{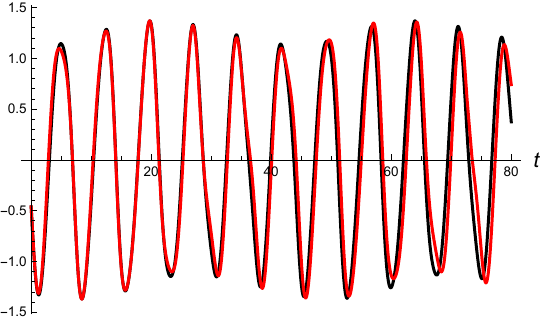}
\caption{Red: $\mathrm{Re}(y_1(t))$ obtained by direct numerical integration of \eqref{ex1:eq} as given in Figure \ref{fig:1}.
Black: $\mathrm{Re}(Y_1(t))$ obtained from the renormalized expansion \eqref{ex1:yy}, using $R(t)$ and $\theta(t)$ in Figure \ref{fig:RT}.}
\end{figure}

\vspace{0.2cm}
\noindent
\section*{Acknowledgments}

A.K. thanks Teiji Kunihiro for his kind interest in the earlier work \cite{K}.
This work was supported by JSPS KAKENHI Grant Number 24K06882.

\let\doi\relax

\end{document}